\documentclass[12pt,letterpaper]{article}

\usepackage[english]{babel}
\usepackage[utf8]{inputenc}
\usepackage{amsmath,amssymb}
\usepackage{graphicx}
\usepackage{subcaption}
\usepackage{hyperref}
\usepackage{makeidx}  

\usepackage[usenames,dvipsnames]{color}

\newtheorem{theorem}{Theorem}
\newtheorem{lemma}{Lemma}

\usepackage[margin=1in]{geometry}
\newenvironment{proof}{\noindent{\it Proof.\,}}{\hfill$\Box$}

\newcommand{\x}{\mathbf{x}}
\newcommand{\bxi}{\boldsymbol{\xi}}
\newcommand{\w}{\mathbf{w}}
\newcommand{\E}{\mathrm{E}}
\newcommand{\Prob}{\mathrm{P}}
\newcommand{\G}{\mathrm{\Gamma}}
\newcommand{\boldeta}{\boldsymbol{\eta}}
\newcommand{\y}{\mathbf{y}}

\title{Preferential placement for community structure formation}

\author{Aleksandr Dorodnykh$^{1}$, Liudmila Ostroumova Prokhorenkova$^{1,2}$, \\ Egor Samosvat$^{1,2}$}

\date{$^1$Moscow Institute of Physics and Technology, Moscow, Russia\\
$^2$Yandex, Moscow, Russia}

\begin{document}


\maketitle

\begin{abstract}
Various models have been recently proposed to reflect and predict different properties of complex networks. However, the community structure, which is one of the most important properties, is not well studied and modeled. In this paper, we suggest a principle called ``preferential placement'', which allows to model a realistic community structure. We provide an extensive empirical analysis of the obtained structure as well as some theoretical results.
\end{abstract}

\section{Introduction}

The evolution of complex networks attracted a lot of attention in recent years. 
Empirical studies of different real-world networks have shown that such structures have some typical properties: small diameter, power-law degree distribution, clustering structure, and others~\cite{boccaletti2006complex,costa2007characterization,newman2003structure}. 
Therefore, numerous random graph models have been proposed to reflect and predict such quantitative and topological aspects of growing real-world networks~\cite{boccaletti2006complex,bollobas2003mathematical,costa2007characterization,ostroumova2016recency,raigorodskii2015small}. 

The most extensively studied property of complex networks is their vertex degree distribution. For the majority of studied real-world networks, the portion of vertices of degree $d$ was observed to decrease as $d^{-\gamma}$, usually with $2 < \gamma < 3$~\cite{barabasi1999emergence,faloutsos1999power,newman2005power}. Such networks are often called scale-free. The most well-known approach to modeling scale-free networks is called \textit{preferential attachment}. The main idea of this approach is that new vertices emerging in a graph connect to some already existing vertices chosen with probabilities proportional to their degrees. Preferential attachment is a natural process allowing to obtain a graph with a power-law degree distribution, and many random graph models are based on this idea, see,~e.g., \cite{bollobas2001degree,buckley2004popularity,holme2002growing,krot2017local,zhou2005maximal}.

Another important characteristic of complex networks is their community (or clustering) structure, i.e., the presence of densely interconnected sets of vertices, which are usually called clusters or communities~\cite{fortunato2010community,girvan2002community}. 
Several empirical studies have shown that community structure of different real-world networks has some typical properties. 
In particular, it was observed that the cumulative community size distribution obeys a power law with some parameter $\lambda$.\footnote{Cumulative community size distribution is the function defined for each $x$ as the probability that a size of a randomly sampled cluster is at least $x$.}
For instance,~\cite{clauset2004finding} reports that $\lambda=1$ for some networks; \cite{arenas2004community} obtains either $\lambda = 0.5$ or $\lambda=1$; \cite{guimera2003self} also observes a power law with $\lambda$ close to 0.5 in some range of cluster sizes; \cite{palla2005uncovering} studies the overlapping communities and shows that $\lambda$ is ranging between 1 and~1.6. 

Community structure is an essential property of complex networks. For example, it highly affects the spreading of infectious diseases in social networks~\cite{hufnagel2004forecast,lipsitch2003transmission},
spread of viruses over computer networks~\cite{wang2000computer}, promotion of products via viral
marketing \cite{kempe2003maximizing},
propagation of information~\cite{romero2011differences}, etc.
Therefore, it is crucial to be able to model realistic community structures.

Unfortunately, the widely used preferential attachment model fails to provide desired clustering structure~\cite{bollobas2003mathematical} and only a few random graph models are able to generate realistic clusters.
Probably the most well-known model was suggested in~\cite{lancichinetti2008benchmark} as a benchmark for comparing community detection algorithms. In this model, the distributions of both degrees and community sizes follow power laws with predetermined exponents. 
However, there are two drawbacks of this model. First, it does not explain the power-law distribution of community sizes, these sizes are just sampled from a power-law distribution at the beginning of the process.
Second, a subgraph induced by each community is very similar to the configuration model~\cite{bender1978asymptotic}, which does not allow to model, e.g., hierarchical community structure often observed in real-world networks~\cite{arenas2004community,clauset2004finding}.

A weighted model which naturally generates communities was proposed in~\cite{kumpula2009model}. However, the community structure in this model is not analyzed in detail and only the local clustering coefficient is shown. From the figures presented in~\cite{kumpula2009model} it seems that the community size distribution does not have a heavy tail as it is observed in real-world complex networks.

Another promising model which naturally generates clusters is proposed in~\cite{bagrow2013natural} and it is based on a so-called \textit{clustering attachment}. However, again, the distribution of cluster sizes was not analyzed for this model. Although from the presented figures and simulations\footnote{\url{http://rocs.hu-berlin.de/interactive/ca/index.html}} it seems that this distribution is not heavy tailed, it would be interesting to analyze it theoretically or empirically.

Let us also mention a paper~\cite{pollner2005preferential} which analyzes a community graph, where vertices refer to communities and edges correspond to shared members between the communities.
The authors show that the development of the community graph seems to be driven by preferential attachment. They also introduce a model for the dynamics of overlapping communities. Note that~\cite{pollner2005preferential} only models the membership of vertices and does not model the underlying network.

Finally, note that a simplified way to analyze the clustering structure is to measure the clustering coefficient, i.e., the probability that two neighbors of a vertex are connected. Several models with high clustering coefficient were proposed in the literature. For example, Spatial Preferential Attachment model~\cite{aiello2008spatial} was proven to have a power-law degree distribution, high clustering coefficient and other desirable properties~\cite{iskhakov2018clustering,prokhorenkova2017modularity}. Similarly, hyperbolic random graphs~\cite{krioukov2010hyperbolic} generate a power-law degree distribution, low diameter and a high clustering coefficient. Also, preferential attachment models were modified by introducing special steps of triangle formation~\cite{holme2002growing}. However, a high clustering coefficient does not necessary lead to a realistic clustering structure. The presence of clusters and the distribution of their sizes were not analyzed in the models mentioned above.

In this paper, we propose a process which naturally generates clustering structures. 
Our approach is called \textit{preferential placement} and it is based on the idea that vertices can be embedded in a multidimensional space of latent features. The vertices appear one by one and their positions are defined according to preferential placement: new vertices are more likely to fall into already dense regions.  
We present a detailed description of this process in Section~\ref{sec:model}. 
After $n$ steps we obtain a set of $n$ vertices placed in a multidimensional space. In Section~\ref{sec:empirical} we empirically and theoretically analyze the obtained structure: in particular, we show that the communities are clearly visible and their sizes are distributed according to a power law. Note that after the placement of all vertices is defined, one can easily construct an underlying network, using, e.g., the threshold model \cite{bradonjic2008structure,masuda2005geographical}. We discuss possible models and their properties in Section~\ref{sec:graph}. 

This paper is a journal version of~\cite{dorodnykh2017preferential}. It contains additional theoretical and empirical results for the vertex configuration obtained by preferential placement. We also discuss some new ideas on generating graphs from this configuration and analyze the properties of the obtained random graph models.



\section{Preferential placement}\label{sec:model}

\subsection{Definition}\label{sec:model_def} 

In this section, we describe the proposed approach which we call \textit{preferential placement}. We assume that all vertices are embedded in $\mathbb{R}^d$ for some $d \ge 1$. One can think that coordinates of this space correspond to latent features of vertices. 
Introducing latent features has recently become a popular approach both in predictive and generative models. These models are known by different names such as latent feature models~\cite{menon2010log,miller2009nonparametric},  matrix factorization models~\cite{artikov2016factorization,dunlavy2011temporal,menon2011link}, spatial models~\cite{aiello2008spatial,barthelemy2003crossover,barthelemy2011spatial}, or geographical models~\cite{bradonjic2008structure,masuda2005geographical}. The basic idea behind all these models is that vertices having similar latent features are more likely to be connected by an edge.

Preferential placement is a procedure describing the embedding of vertices in the space $\mathbb{R}^d$. After that, given the coordinates of all vertices, one can construct a graph using one of many well-known approaches (see Section~\ref{sec:graph} for the discussion of possible variants).

Our model is parametrized by a distribution $\Xi$ taking nonnegative values. The proper choice of $\Xi$ is discussed further in this section.

We construct a random configuration of vertices (or points) $S_n = \{\x_1, \ldots, \x_n\}$, where $\x_i = (x_i^1, \ldots, x_i^d)$  denotes the coordinates of the $i$-th vertex $v_i$. Let $S_1 = \{\x_1\}$, $\x_1$ is the origin. Now assume that we have constructed $S_t$ for $t\ge1$, then we obtain $S_{t+1}$ by adding a vertex $v_{t+1}$ with the coordinates $\x_{t+1}$ chosen in the following way:
\begin{itemize} 
\item Choose a vertex $v_{i_{t+1}}$ from $v_1, \ldots, v_t$ uniformly at random.
\item Sample $\xi_{t+1}$ from the distribution $\Xi$.
\item Sample a direction $\mathbf{e}_{t+1}$ from a uniform distribution on a multidimensional sphere $\|\mathbf{e}_{t+1}\| = 1$, where by $\|\cdot\|$ we denote the Euclidean norm in $\mathbb{R}^d$. 
\item Set $\x_{t+1} = \x_{i_{t+1}} + \xi_{t+1} \cdot \mathbf{e}_{t+1}$.                           
\end{itemize} 
We argue in this paper that in order to obtain a realistic clustering structure one should take $\Xi$ to be a heavy tailed distribution. In this case, according to the procedure described above, new vertices will usually appear in the dense regions, close to some previously added vertices; however, due to the heavy tail of $\Xi$, from time to time we get outliers, which originate new clusters.

We call the described above procedure ``preferential placement'' due to its analogy with preferential attachment. Assume that at some step of the algorithm we have several clusters, i.e., groups of vertices located close to each other, and a new vertex appears. Then the probability that this vertex will join a cluster $C$ is roughly proportional to its size, i.e., the number of vertices already belonging to this cluster.
This is the basic intuition which we discuss further in this paper in more details.

\subsection{Tree-based interpretation}\label{sec:model_tree}  

In some cases, it is convenient to think of the definition given above in the following way.
Let us first construct a uniform random recursive tree. 
Recall that a recursive tree of order $n$ is a rooted tree on $n$ vertices labeled $v_1, \ldots, v_n$, with the property that for each $k$, $2 \le k \le n$, the indices of the vertices on the unique path from the root $v_1$ to the vertex $v_k$ form an increasing sequence~\cite{smythe1995survey}. 
Uniform random recursive tree (or URRT) is a tree sampled from the set of recursive trees of order $n$ uniformly at random. An equivalent way to sample a URRT is to start from a graph $T_1$ consisting of a root vertex $v_1$ and at each step $t+1$ add a vertex $v_{t+1}$ connected to a vertex $v_{i_{t+1}}$ chosen from $v_1, \ldots, v_t$ uniformly at random. After $n$ steps we get a URRT $T_n$. Random recursive trees were extensively studied in the literature, see, e.g.,~\cite{devroye1995strong,dobrow1999total,najock1982number,pittel1994note}.

To construct a random configuration of vertices $S_n = \{\x_1, \ldots, \x_n\}$, we first construct a URRT $T_n$ using a recursive procedure described above (we will further call this tree \textit{genealogical}). Then we label all edges of $T_n$: for each edge $(v_{t},v_{i_t})$ we set its label $\w_t$ to be a vector $\xi_t \mathbf{e}_{t}$, where $\xi_t$ is sampled from the distribution $\Xi$ and $\mathbf{e}_{t}$ is, as before, a random vector of length 1 in $\mathbb{R}^d$. Finally, the value $\x_k$ for each vertex $v_k$ is obtained by summing all labels on the unique path from $v_1$ to $v_k$, i.e., $\x_k = \w_{j_1}+ \ldots + \w_{j_l}$ such that for all $1 \le t \le l$ we have $i_{j_t} = j_{t-1}$, $j_l = k$, $j_0 = 1$. It is easy to see that this definition is equivalent to the one given in Section~\ref{sec:model_def}. 

Note that each $\x_k = \w_{j_1}+ \ldots + \w_{j_l}$ is obtained by a random walk with lengths of jumps distributed according to $\Xi$, which we further assume to be a power law. Such random walks are known in the literature as L{\'e}vy flights~\cite{chechkin2008introduction}. In other words, preferential placement is essentially a combination of URRT and L{\'e}vy flights.





\section{Analysis of preferential placement}\label{sec:empirical}

\subsection{Experimental setup}\label{sec:setup}

In this section, we analyze structures obtained using the preferential placement procedure described above. We take $\Xi$ to be a slightly modified Pareto distribution with the density function $f_\beta(x) = \frac \beta {(x + 1)^{\beta+1}}, x\geq 0$ for fixed $\beta > 0$. 
This distribution is supported on the interval $[0,+\infty)$ and is asymptotically equivalent to the Pareto distribution.

In all the experiments we take $d=2$ since the obtained structures are easy to visualize. However, we also tried other values of $d\ge 1$ and obtained results similar to shown, e.g., on Figures~\ref{fig:size_distr_all},~\ref{fig:visualization} and~\ref{fig:degree_distr}.
Also, if not specified otherwise, we generated structures with the number of points $n = 10^5$.
In Section~\ref{sec:parameters}, we analyze the effect of $d$, $n$ and $\beta$ on the distribution of cluster sizes.


\begin{figure*}[t]
        \centering
        \begin{subfigure}[b]{0.49\textwidth}
            \centering       \includegraphics[width=\textwidth]{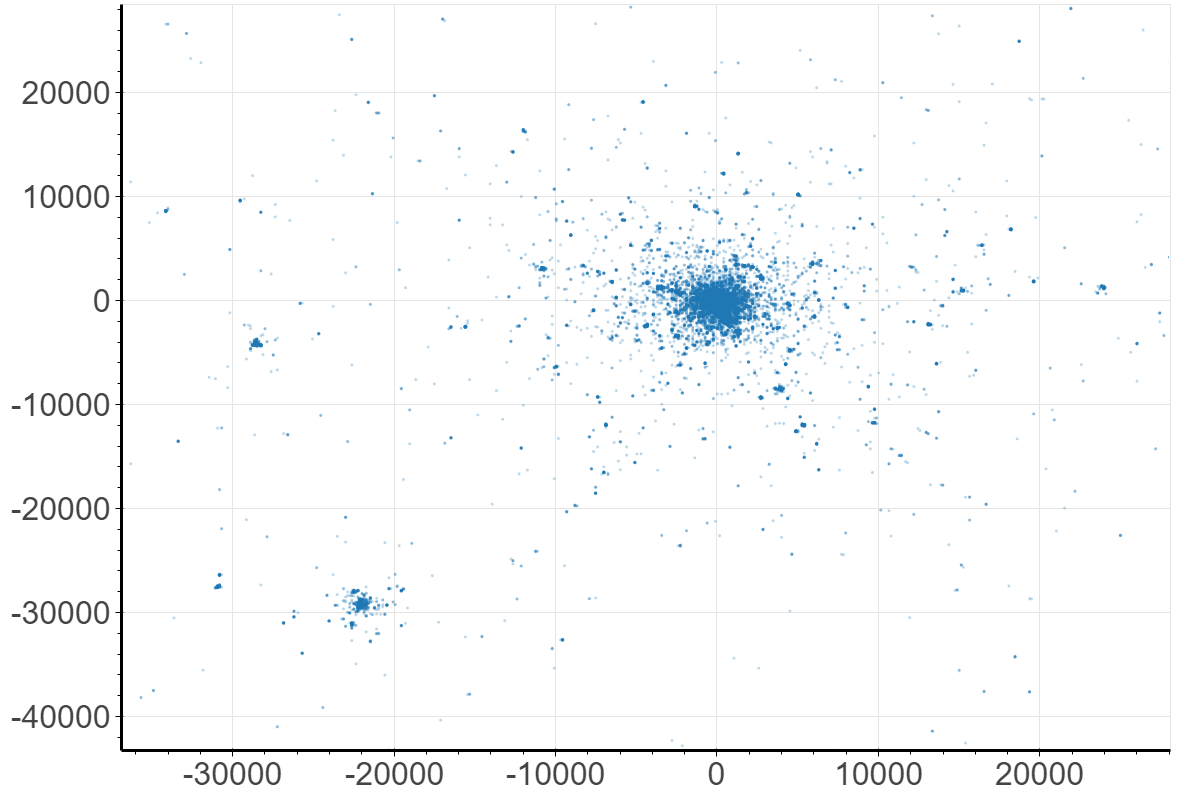}
 \caption{$\beta=0.5$}\label{fig:beta=0.5}
        \end{subfigure}
        \begin{subfigure}[b]{0.49\textwidth}
            \centering            \includegraphics[width=\textwidth]{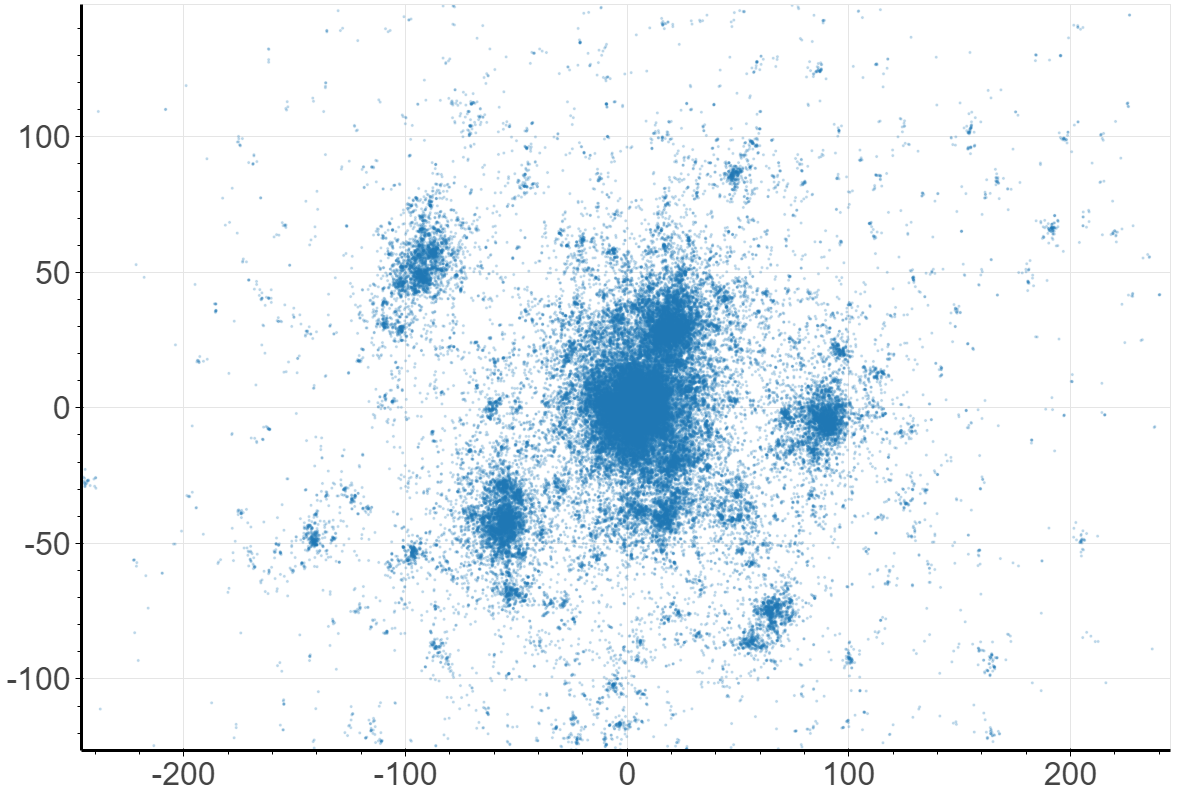}
            \caption{$\beta=1$}
        \end{subfigure}
        \begin{subfigure}[b]{0.49\textwidth}
            \centering
 \includegraphics[width=\textwidth]{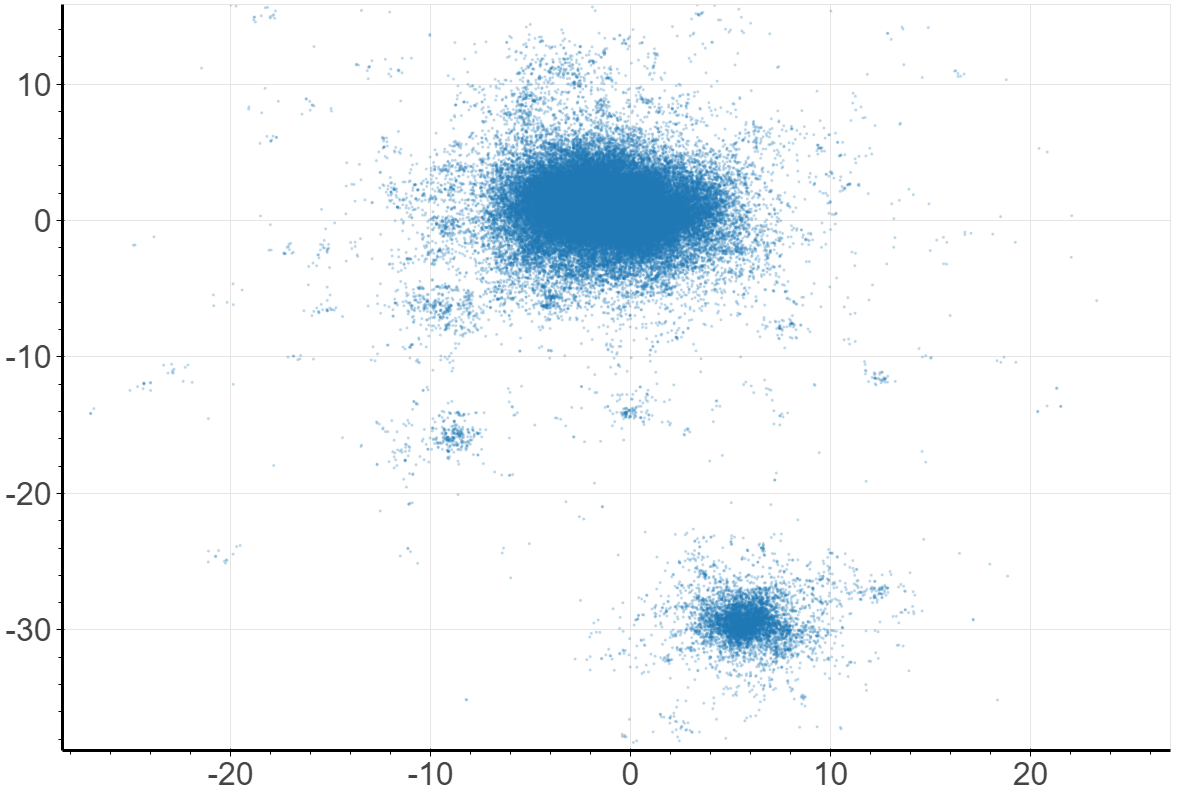}            \caption{$\beta=2.5$}
        \end{subfigure}
        \begin{subfigure}[b]{0.49\textwidth}
            \centering
 \includegraphics[width=\textwidth]{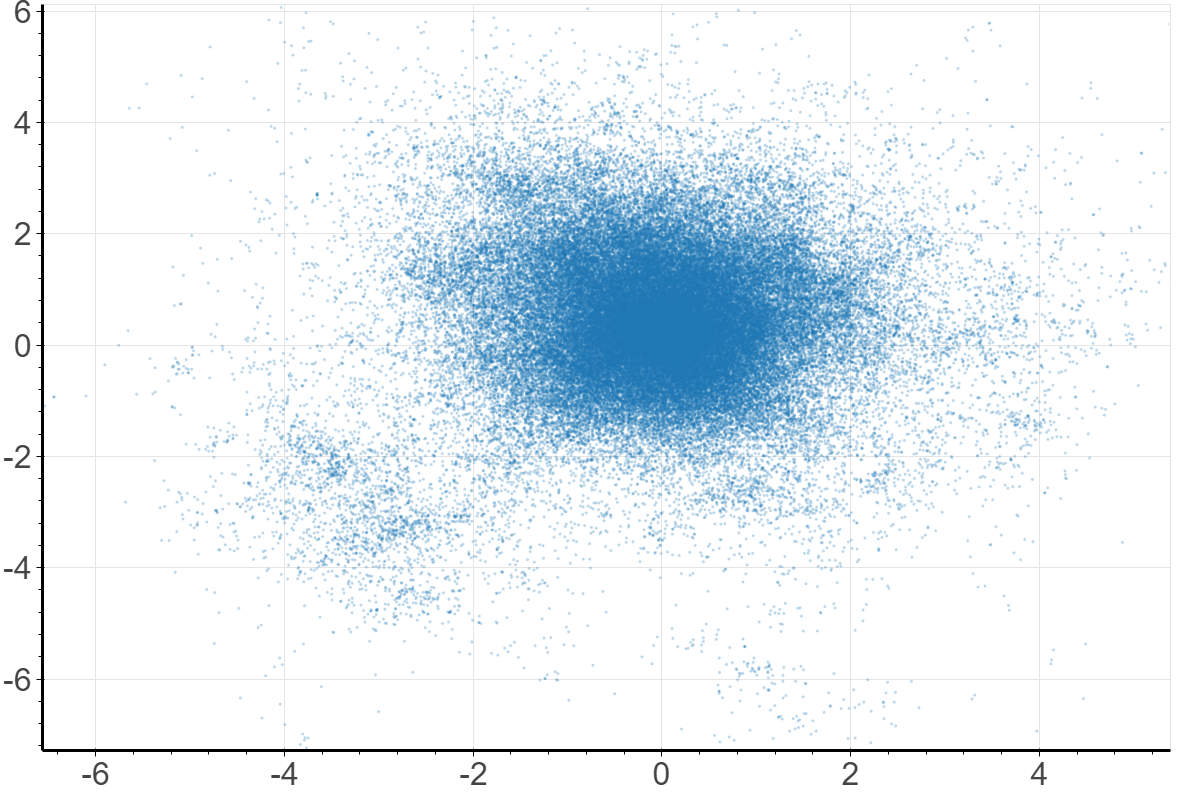}            \caption{$\beta=4$}\label{fig:beta=4}
        \end{subfigure}
        \caption{Clustering structure depending on $\Xi$.}
        \label{fig:beta}
\end{figure*}

\subsection{Clustering structure depending on $\Xi$}

First, let us visualize the structures obtained by our algorithm. We tried several values of $\beta$, $\beta \in \{0.5, 1, 1.5, 2.5, 4\}$. 
The results are presented on Figures~\ref{fig:beta} and~\ref{fig:beta_scale}. 
The value $\beta = 0.5$ produces the heaviest tail, in this case the distribution $\Xi$ does not have a finite expectation. Although some clusters are clearly visible in this case (Figure~\ref{fig:beta=0.5}), they are located far apart from each other
and there are too many single outliers laying far away from other points, which seems to be not very realistic.
Note that for too large $\beta$, e.g., for $\beta = 4$, the variance is too low and we obtain only one giant cluster with minor fluctuations, as presented on Figure~\ref{fig:beta=4}.
Further in this paper we discuss the case $\beta = 1.5$ presented on Figure~\ref{fig:beta_scale}. In this case  $\Xi$ has a finite expectation but an infinite variance.

\begin{figure*}[t]
        \centering
        \begin{subfigure}[b]{0.49\textwidth}
            \centering
            \includegraphics[width=\textwidth]{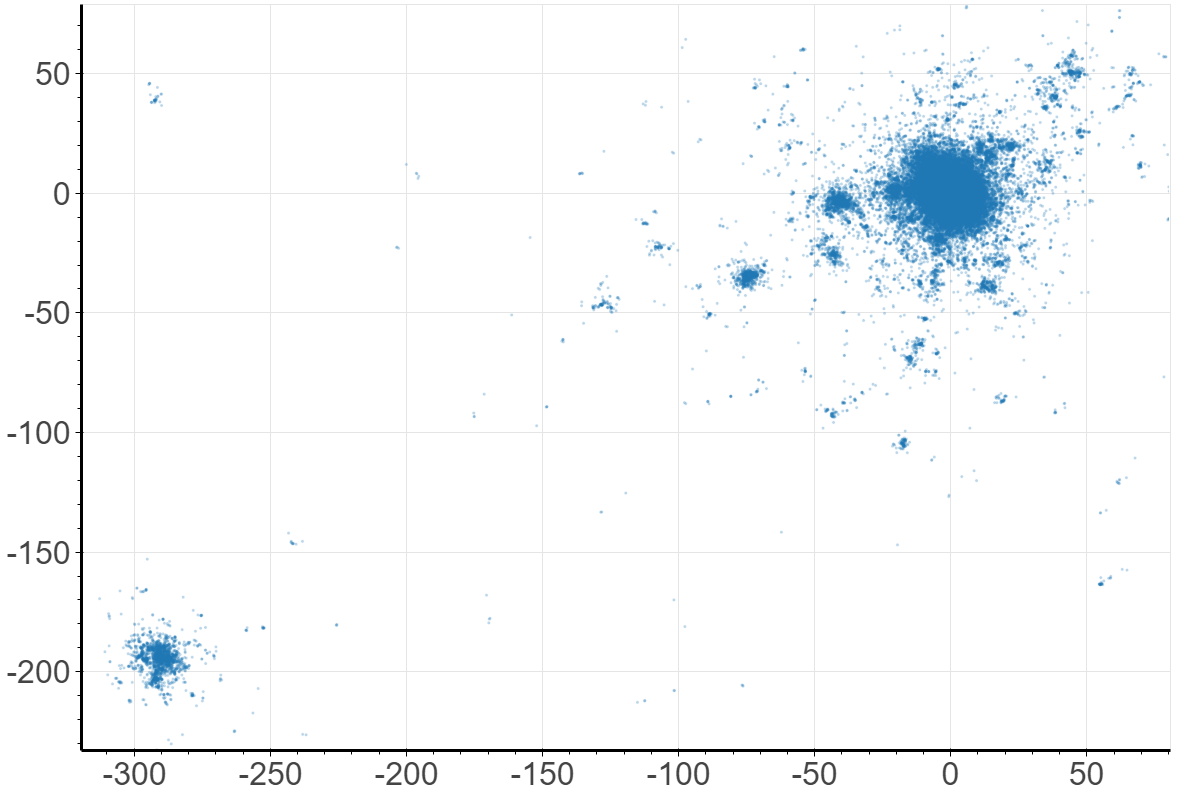}
            \label{fig:beta_scale_1}
        \end{subfigure}
        \begin{subfigure}[b]{0.49\textwidth}
            \centering
            \includegraphics[width=\textwidth]{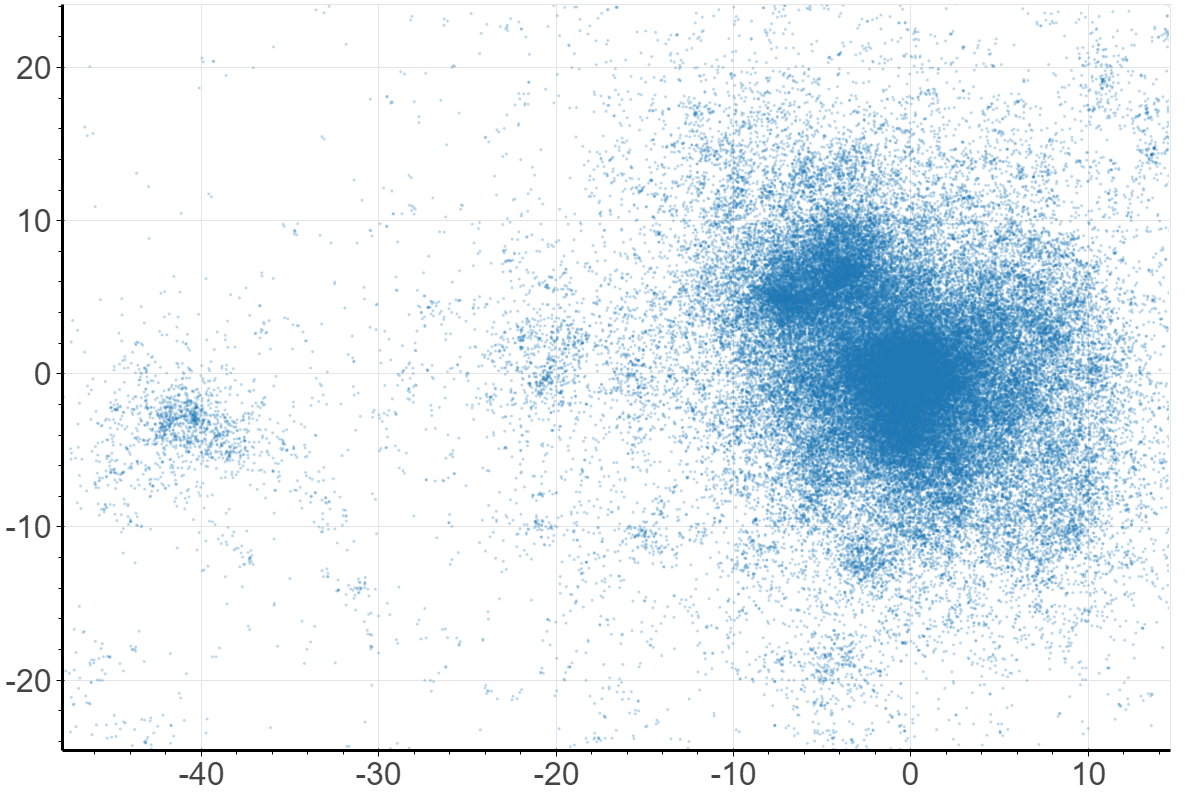}
            \label{fig:beta_scale_2}
        \end{subfigure}
        \caption{$\beta = 1.5$, different scales.}
        \label{fig:beta_scale}
\end{figure*}

Another interesting observation is a hierarchical clustering structure produced by our algorithm. To illustrate this, we take the figure obtained for $\beta = 1.5$ and zoom it to see more details. Figure~\ref{fig:beta_scale} shows that the largest cluster further consists of several sub-clusters.

In order to further analyze this self-similarity property, we empirically estimated a fractal dimension of the obtained structures~\cite{kenneth2003fractal}. 
A fractal dimension is an index characterizing how a fractal pattern changes with the scale at which it is measured.
Unlike topological dimensions, the fractal dimension can take non-integer values. For example, a curve with a fractal dimension close to 1 behaves quite like an ordinary line, but a curve with a fractal dimension 1.9 winds convolutely through space nearly like a surface.
We used a box-counting algorithm\footnote{\url{https://en.wikipedia.org/wiki/Box_counting}} which computes the number of boxes $N(\varepsilon)$ with the edge size equal to $\varepsilon$ needed to cover the structure. Then the fractal dimension $d_f$ is defined using the following relation: $N(\varepsilon) \propto \varepsilon^{-d_f}$. The results are presented on Figure~\ref{fig:fd} (crosses represent the empirically observed pairs $(\varepsilon,N(\varepsilon))$, while lines correspond to the approximations $N(\varepsilon) \propto \varepsilon^{-d_f}$). One can see that the fractal dimension $d_f$ increases from~$0.4$~to~$1.4$ while $\beta$ increases from~$0.5$~to~$4$.
Note that the fractal dimension of L{\'e}vy flights is also known to monotonically increase with the parameter of the power-law distribution~\cite{laskin2000fractional}.

\begin{figure*}[t]
        \centering
        \begin{subfigure}[b]{0.49\textwidth}
            \centering
            \includegraphics[width=\textwidth]{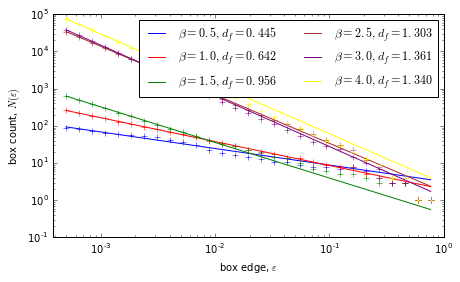}
        \end{subfigure}
        \begin{subfigure}[b]{0.49\textwidth}
            \centering
            \includegraphics[width=\textwidth]{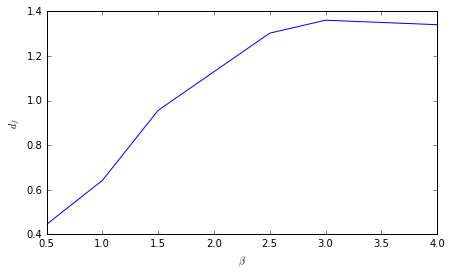}
        \end{subfigure}
        \caption{Fractal dimension $d_f$.}
        \label{fig:fd}
\end{figure*}

\subsection{The distribution of cluster sizes}

In this section, we analyze the distribution of cluster sizes produced by preferential placement. We present both theoretical and empirical observations.

The main difficulty with the analysis of clustering structures is the fact that there are no standard definitions of clusters, both in graphs and metric spaces.
For example, clusters are often defined as a result of some clustering algorithm.\footnote{\textit{Modularity}, introduced in~\cite{newman2004finding}, can be used to define communities in graphs. However, this characteristic has certain drawbacks, as discussed in~\cite{fortunato2007resolution}. Moreover, modularity favors partitions with approximately equal communities, which contradicts the main idea of power-law distribution of community sizes. }
This causes a lot of difficulties for both theoretical and empirical analysis.

\subsubsection{Theoretical analysis}\label{sec:theory}

First, let us discuss why we expect to observe a power-law distribution of cluster sizes in our model.
As we discussed above, due to the absence of a rigorous definition of a cluster, further in this section we are only able to give some insights to the fact that the proposed algorithm is expected to generate power-law distribution of cluster sizes. Namely, we make some strong assumptions and then rigorously prove that the distribution follows a power law.

Let $F_t(s)$ denote the number of clusters of size $s$ at step $t$. In order to analyze $F_t(s)$ we consider its dynamics inductively. 
Assume that after a step $t$ we obtain some clustering structure.
At step $t+1$ we add a vertex $v_{t+1}$ and choose its ``parent'' $v_{i_{t+1}}$ from $v_1, \ldots, v_t$ uniformly at random.
Clearly, the probability to choose a parent from some cluster $C$ with $|C| = s$ is equal to $\frac{s}{t}$. In this case, we call $C$ a parent cluster for $v_{t+1}$.
Now let us make the following assumptions:
\begin{enumerate}
\item All clusters can only grow, they cannot merge or split. 
\item At step $t+1$ a new cluster appears with probability $p(t) = \frac{c}{t^{\alpha}}$, $c>0$, $0\le\alpha \le 1$.\footnote{For simplicity, we choose $c$ such that $p(t) < 1$ for all $t\ge 2$. Although our proof can be extended to cases where $p(t)$ behaves as $\frac{c}{t^\alpha}$ only asymptotically (which would allow to take larger $c$), we do not formally consider this case.}
\item Given that a vertex $t+1$ does not create a new cluster, the probability to join a cluster $C$ with $|C| = s$ is equal to $\frac{s}{t}$.
\end{enumerate}

These assumptions are quite strong and even not very realistic. For instance, it seems reasonable that two clusters can merge if many vertices appear somewhere between them. In Section~\ref{sec:empirical_analysis} we discuss how the violation of the first assumption by the largest clusters affects the observed distribution of cluster sizes. Regarding the second assumption, $p(t)$ can possibly depend on the current configuration $S_{t}$. However, these assumptions allow us to analyze the behavior of $F_t(s)$ formally. Namely, we prove the following theorem.

\begin{theorem}\label{th:cluster_sizes}
Under the assumptions described above the following holds.
\begin{enumerate}
\item If $\alpha=0$ and $0<c<1$, then
$$
\E F_n(s) = \frac{c\, (s-1)!\,\G\left(2+\frac{1}{1-c}\right)}{(2-c)\G\left(s+1+\frac{1}{1-c}\right)}\left(n + O\left(s^\frac{1}{1-c}\right)\right)
\sim \frac{c \,\G\left(2+\frac{1}{1-c}\right)}{(2-c)}\cdot \frac{n}{s^{1+\frac{1}{1-c}}} \,.
$$
\item If $0<\alpha\le 1$, then for any $\epsilon>0$
$$
\E F_n(s) = \frac{c \,(s-1)! \, \G(3-\alpha)}{(2-\alpha)\G(s+2-\alpha)} \left(n^{1-\alpha} + O\left(n^{\max\{0,1-2\alpha\}}s^{1-\alpha+\epsilon}\right)\right) 
\sim \frac{c \, \G(3-\alpha)}{2-\alpha} \cdot \frac{n^{1-\alpha}}{s^{2-\alpha}} \,.
$$
\end{enumerate}
\end{theorem}

To sum up, if the probability $p(n)$ of creating a new cluster is of order $\frac{1}{n^{\alpha}}$ for $\alpha>0$, then the distribution of cluster sizes follows a power law with parameter $2-\alpha$ growing with $p(n)$ from 1 to 2; if $p(n) = c$, $0<c<1$, then the parameter grows with $c$ from 2 to infinity. Note that for power-law distributions, the parameter of the cumulative distribution $\lambda$ is one less than the parameter of the corresponding probability mass function, i.e., the exponent in the theorem discussed above is $\lambda+1$.
The proof of Theorem~\ref{th:cluster_sizes} is technical and we place it in the Appendix.

 
Let us also explain why we do not consider $p(n)$ decreasing faster than $\frac{c}{n}$.
It is natural to assume that a new cluster appears if a new vertex chooses a parent vertex near the border of some cluster and then $\xi_{t+1}$ and $\mathbf{e}_{t+1}$ are chosen such that $\x_{t+1} = \x_{i_{t+1}} + \xi_{t+1} \cdot \mathbf{e}_{t+1}$ falls quite away from the parent cluster. This probability is roughly proportional to the number of vertices located near the borders of the clusters. For each cluster, at least one vertex has to be near its border, so we get the lower bound $\frac{c}{n}$.

Finally, let us mention that in practice the probability $p(n)$ of creating a new cluster can depend not only on $\Xi$, but also on the definition of clusters. Further in this section we demonstrate that parameters of a clustering algorithm can affect the parameter of the obtained power law.

\begin{figure*}[t]
        \centering
        \begin{subfigure}[b]{0.49\textwidth}
            \centering
            \includegraphics[width=\textwidth]{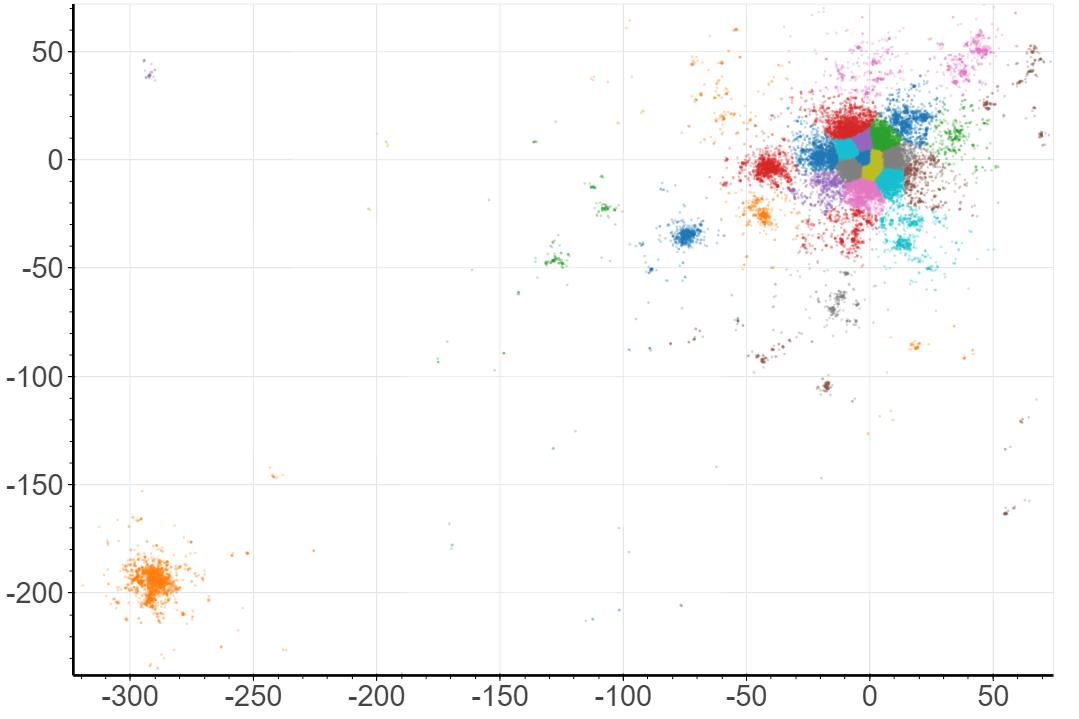}
        \caption{k-means, $k=50$}\label{fig:k-means}
        \end{subfigure}
        \begin{subfigure}[b]{0.49\textwidth}
            \centering
        \includegraphics[width=\textwidth]{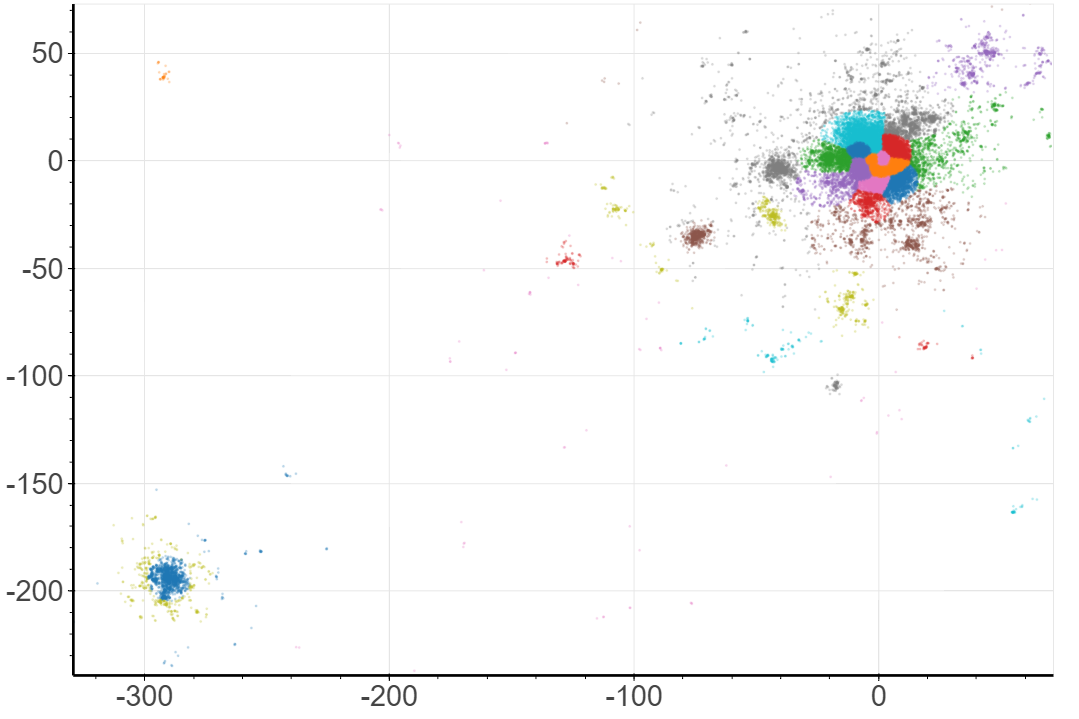}
         \caption{EM, $k = 50$}\label{fig:em}
        \end{subfigure}
           \begin{subfigure}[b]{0.49\textwidth}
            \centering
           \vspace{0.5cm}
            \includegraphics[width=\textwidth]{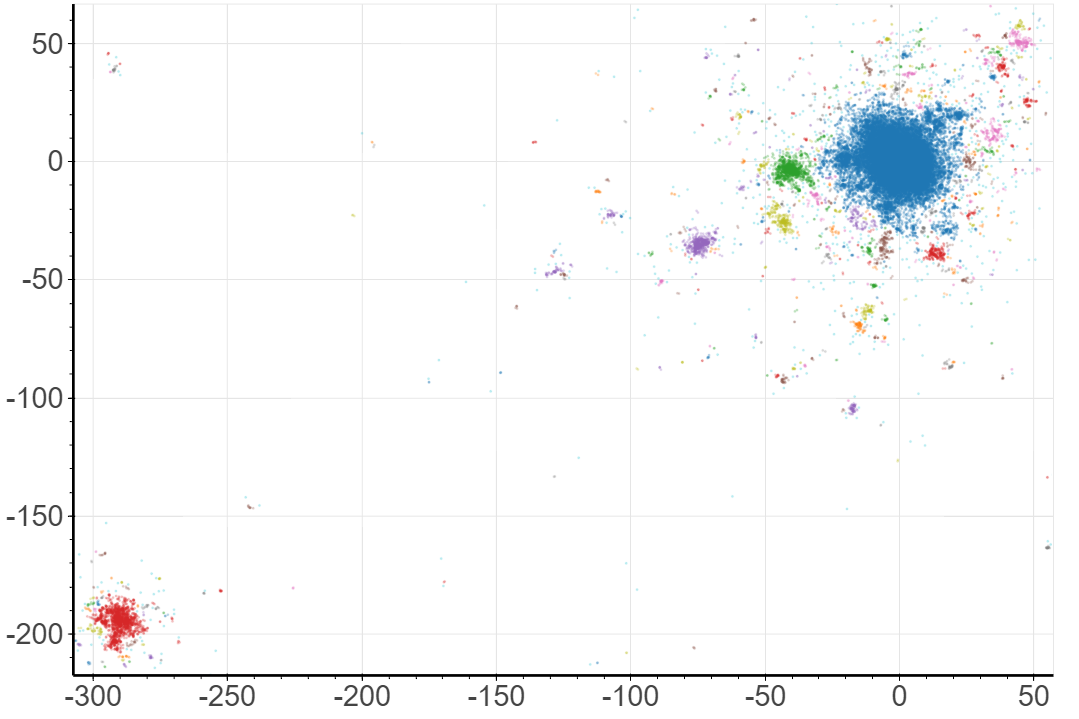}
         \caption{DBSCAN, $L = 125$, $l=1$}\label{fig:dbscan}
        \end{subfigure}
        \caption{The comparison of different clustering algorithms.}
        \label{fig:cluster_algorithms}
\end{figure*}

\subsubsection{Empirical analysis.}\label{sec:empirical_analysis}

\begin{figure*}[t]
        \centering
             \begin{subfigure}[b]{0.49\textwidth}
            \centering
            \includegraphics[width=\textwidth]{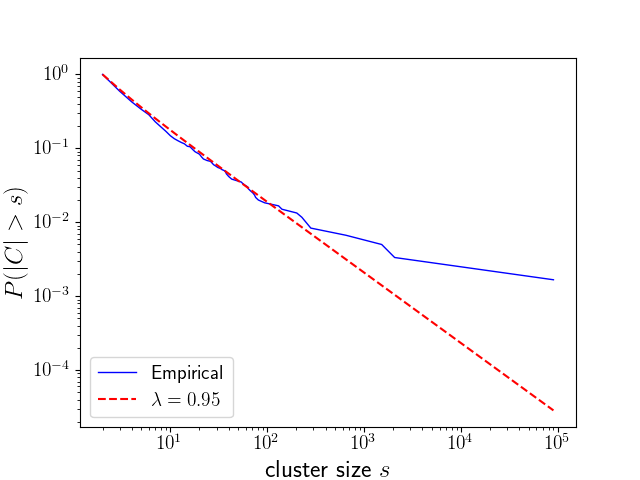}
         \caption{DBSCAN with $L=125,l=1$}\label{fig:size_distr}
        \end{subfigure}
        \begin{subfigure}[b]{0.49\textwidth}
            \centering
            \includegraphics[width=\textwidth]{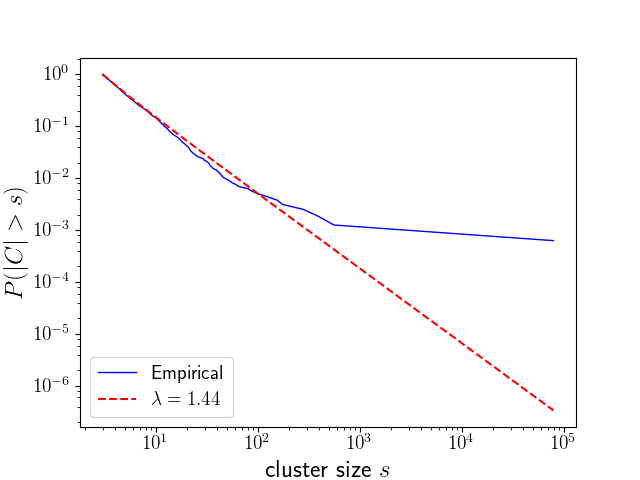}
         \caption{DBSCAN with $L=5,l=1$}\label{fig:size_distr_1}
        \end{subfigure}
          \caption{Cluster size distribution.}\label{fig:size_distr_all}
\end{figure*}
\begin{figure*}[h!]
            \centering
       \includegraphics[width=0.7\textwidth]{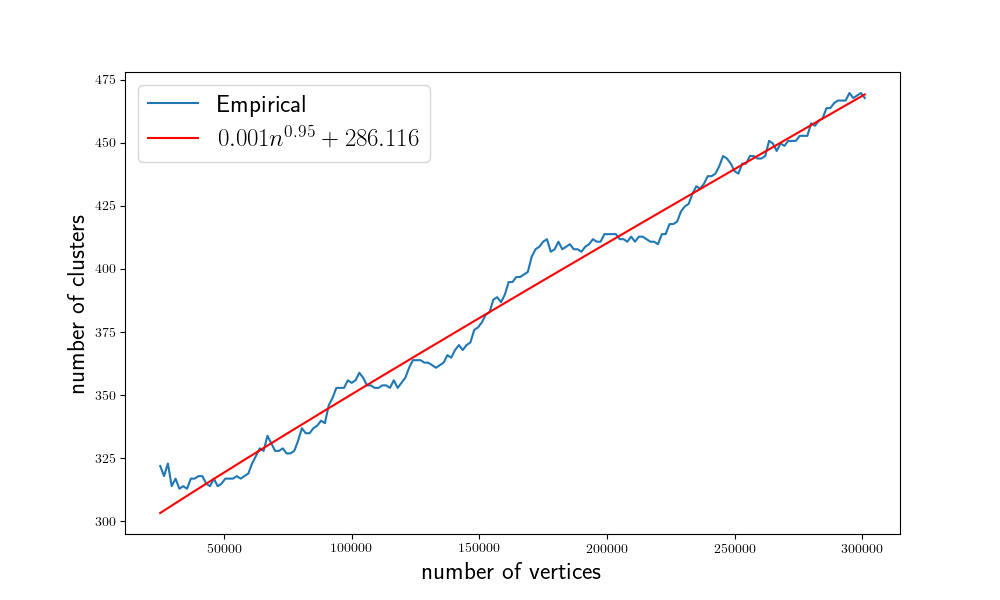}
        \caption{Growth of the number of clusters, DBSCAN with $L=125,l=1$.}\label{fig:number_of_clusters}
\end{figure*}

As we already mentioned, there is no standard definition of a clustering structure. In many cases, clusters and communities are defined just as a result of some clustering algorithm. Therefore, we first analyze the performance of several clustering algorithms, then choose the most appropriate one and analyze clusters it produces.

We compare the following algorithms: k-means~\cite{lloyd1982least}, EM (expectation maximization), and DBSCAN (density-based spatial clustering of applications with noise) \cite{ester1996density}.
For k-means and EM one has to specify the number of clusters. We tried several values of $k$, $k \in \{10, 50, 100, 500, 1000\}$, but both algorithms turned out to be not suitable for our problem. As expected, in all cases they unnaturally split the largest cluster into several small ones (see Figures~\ref{fig:k-means} and~\ref{fig:em}).

On the contrary, DBSCAN  produces more realistic results. It requires two parameters: radius of neighborhood $\varepsilon$ and the minimum number of neighbors required to form a dense region $l$. We consider $l \in \{1,2,3\}$ and $\varepsilon$ is chosen in such a way that if we connect all pairs of vertices $i,j$ such that $\|i - j\| < \varepsilon$, then we get $Ln$ edges, $L \in \{5, 25, 125\}$, where $n$ is the number of vertices. For all parameters we get reasonable clustering structures. The result for $L=125, l=1$ is presented on Figure~\ref{fig:dbscan}. For these parameters we also analyze the distribution of cluster sizes (see Figure~\ref{fig:size_distr}). Note that for not too large values of $s$ ($s<300$) the cumulative distribution follows a power law with parameter $\lambda \approx 0.95$. In Theorem~\ref{th:cluster_sizes} this value corresponds to the case $\alpha = 0.05$, i.e., $p(n) \propto n^{-0.05}$. Based on this, we expect the number of clusters to grow as $n^{0.95}$, i.e., close to linearly. In Figure~\ref{fig:number_of_clusters} we plot the empirical number of clusters and fit it by $n^{0.95}$.

Now, as we promised above, we show that $\lambda$ can depend on the clustering algorithm. 
Figure~\ref{fig:size_distr_1} shows the cumulative cluster size distribution for DBSCAN with $L=5,l=1$. We obtain $\lambda=1.44$, so it is larger in this case. Intuitively, the reason is that $p(n)$ is larger for $L=5$ than for $L=125$. Smaller values of $L$ correspond to smaller $\varepsilon$, which means that it is harder for a new vertex to join some existing cluster, which makes $p(n)$ larger.

Finally, let us also discuss the bend in the distribution observed on Figure~\ref{fig:size_distr_all}. Namely, one can clearly see that several largest cluster sizes do not fit the desired power-law curve. This phenomenon is especially prominent for the largest cluster. The possible reason of this bend is the fact that in reality our first assumption in Section~\ref{sec:theory} is violated. Namely, we assumed that clusters cannot merge. However, Figure~\ref{fig:dbscan} clearly shows that the largest cluster further consists of several smaller ones (which is a desired hierarchical community structure). At previous steps of the construction process, these smaller clusters were separated, but then they merged. This fact caused the largest cluster to be larger than predicted, which we can see on Figure~\ref{fig:size_distr_all}. Note that such phenomenon is expected only for a few largest clusters.

\subsubsection{Influence of parameters}\label{sec:parameters}

\begin{figure*}[t]
        \centering
        \begin{subfigure}[b]{0.49\textwidth}
            \centering
            \includegraphics[width=\textwidth]{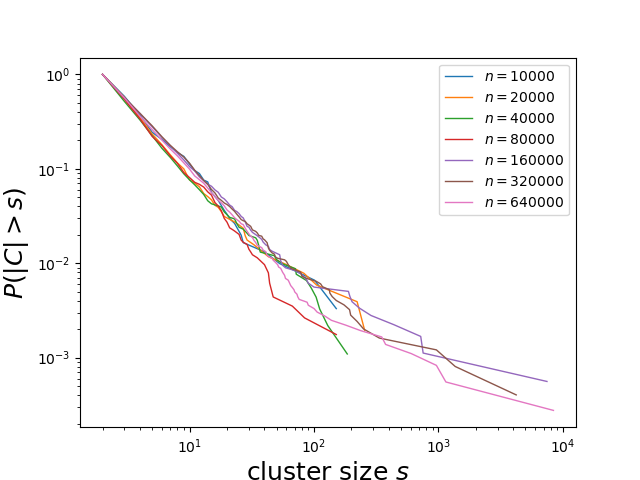}
        \caption{Number of vertices $n$}\label{fig:sizes}
        \end{subfigure}
        \begin{subfigure}[b]{0.49\textwidth}
            \centering
        \includegraphics[width=\textwidth]{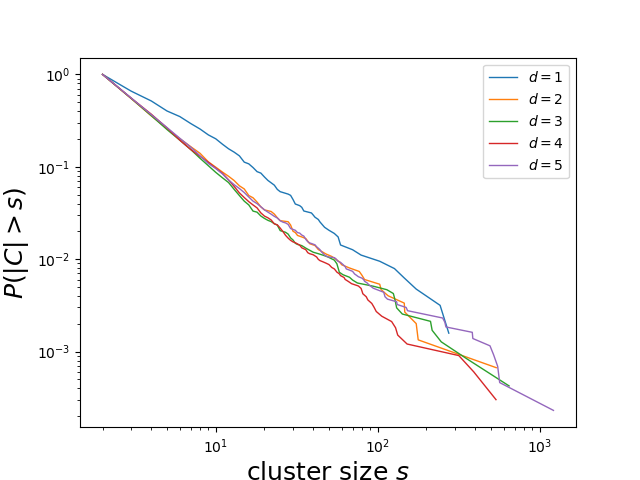}
         \caption{Dimension $d$}\label{fig:dims}
        \end{subfigure}
           \begin{subfigure}[b]{0.49\textwidth}
            \centering
           \vspace{0.5cm}
            \includegraphics[width=\textwidth]{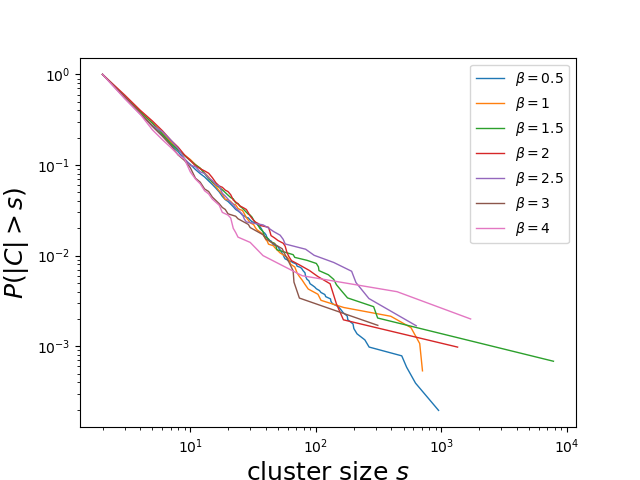}
         \caption{Pareto parameter $\beta$}\label{fig:betas}
        \end{subfigure}
        \caption{Influence of parameters on the distribution of cluster sizes.}\label{fig:parameters}
\end{figure*}

In this section, we discuss how the distribution of cluster sizes is affected by the parameters of the model. In particular, we show that our conclusions hold for various values of the number of vertices $n$, the parameter of Pareto distribution $\beta$, and the dimension $d$. 

As a starting point, we take already discussed parameters $n=10^5$, $d=2$, $\beta = 1.5$ and independently vary each parameter. The obtained results are presented in Figure~\ref{fig:parameters}.\footnote{According to the discussion above, for all curves, we removed the last point, which is an outlier corresponding to the largest cluster.} We noticed that the distribution of cluster sizes is stable and is not affected much by the parameters of the model. The most different curve corresponds to the simple one-dimensional case $d=1$.  

\subsection{Scatter of configurations}\label{sec:radius}

In this section, we analyze the scatter of vertices in $\mathbb{R}^d$. This analysis is motivated by two factors. First,  the scatter of points in $\mathbb{R}^d$ can be considered as an analogue of the diameter of graphs, which is an important and widely studied characteristic.
Second, preferential placement is a simple and elegant process and its theoretical properties are of interest by themselves. Here we were partially motivated by the study of L{\'e}vy flights~\cite{chechkin2008introduction}: it is known that for L{\'e}vy flights we have $\langle |x|^{\delta }\rangle \propto t^{\delta /\beta }{\text{ if }}\delta \leq \beta$. Below we prove a similar result for preferential placement.

Recall that we assume the distribution $\Xi$ to have a density function $f_\beta(x) = \frac \beta {(x + 1)^{\beta+1}}$, ${x\geq 0}$, ${\beta > 0}$.
Let us first consider the variable
$$
X_{\delta}(n) = \frac 1 n \sum_{i=1}^n  \|\x_i\|^\delta, \,\,\, \delta < \beta\,,
$$ 
which is the average moment of order $\delta$ for the distance from the origin to a vertex. 
Recall that the most interesting case is $1 < \beta < 2$, for such $\beta$ the expected average squared distance $\E X_{2}(n)$ diverges, therefore we analyze $\E X_{\delta}(n)$ for $1 \le \delta \le 2$. The following theorem holds.

\begin{theorem}\label{thm:moment}
Let $\xi$ be a random variable sampled from the distribution $\Xi$. 
For any $\delta$, $1 \le \delta < \min\{\beta,2\}$:
$$
\E X_{\delta}(n) \le \E \xi^\delta \sum_{t = 2}^n \frac 1 t \sim \E \xi^\delta \log n\,.
$$
If $\beta > 2$, then 
$$
\E X_{2}(n) = \E \xi^2 \sum_{t = 2}^n \frac 1 t \sim \E \xi^2 \log n\,.  
$$
\end{theorem}

Note that in contrast to L{\'e}vy flights, preferential placement is much more concentrated near the origin (on average). Namely, $X_{\delta}(n)$ grows at most logarithmically with $n$ (instead of $n^{\delta/\beta}$). The reason is that preferential placement makes jumps from random vertices, while L{\'e}vy flight always chooses the latest one.

In the proof of Theorem~\ref{thm:moment} we will use the following lemma.

\begin{lemma}\label{lem:moments}
	Let $\bxi$ and $\boldeta$ be two independent random vectors. Assume that the distribution of $\bxi$ is symmetrical, i.e., the distributions of $\bxi$ and $-\bxi$ are identical. Then for any $1 \le \delta \le 2$:
	$$
	\E\|\bxi + \boldeta\|^\delta \le \E \|\bxi\|^\delta + \E \|\boldeta\|^\delta \,.
	$$
\end{lemma}

We place the proof of this lemma to Appendix. Now we are ready to prove Theorem~\ref{thm:moment}.
\vspace{0.2cm}

\begin{proof}
Let us denote by $\bxi$ a vector of length $\xi$ (sampled from the distribution $\Xi$) pointing to a random direction, i.e., $\bxi = \xi \mathbf{e}$.

For $\beta > 2$ the following recurrent formula holds:
\begin{multline*}
\E X_2(n+1) = \frac{n}{n+1} \E X_2(n) + \frac{1}{n+1} \cdot \frac{1}{n} \cdot \sum_{i=1}^n \E \|\x_i+\bxi\|^2 \\ = \frac{n}{n+1} \E X_2(n) + \frac{1}{n+1} \cdot \frac{1}{n} \cdot \sum_{i=1}^n \E (\|\x_i\|^2+\|\bxi\|^2 + 2 \, \x_i \cdot \bxi) = \E X_2(n) + \frac{1}{n+1} \E \|\bxi\|^2\,.
\end{multline*}
Using the above recurrent formula and the equality $\E X_2(1) = 0$, we prove the second part of the theorem:
$$
\E X_2(n) = \E \|\bxi\|^2 \cdot \sum_{t=2}^n \frac{1}{t} \sim \E \xi^2 \cdot \log n\,.
$$

To prove the first part we use Lemma~\ref{lem:moments}:
\begin{multline*}
\E X_\delta(n+1) = \frac{n}{n+1} \E X_\delta(n) + \frac{1}{n+1} \cdot \frac{1}{n} \cdot \sum_{i=1}^n \E \|\x_i+\bxi\|^\delta
\le \frac{n}{n+1} \E X_\delta(n)  \\ + \frac{1}{n+1} \cdot \frac{1}{n} \cdot \sum_{i=1}^n \E \|\x_i\|^\delta + \frac{1}{n+1} \cdot \frac{1}{n} \cdot \sum_{i=1}^n \E \|\bxi\|^\delta
= \E X_\delta(n) + \frac{1}{n+1} \E \|\bxi\|^\delta \,,
\end{multline*}
so we obtain the recurrent formula for $\E X_{\delta}(n)$ from which the theorem follows. 
\end{proof}

\vspace{0.2cm}

In addition to $X_\delta(n)$, we analyze the following variable:
$$
X_{max}(n) = \max_{1\le i \le n} \|\x_i\|\,.
$$
Note that both $X_{\delta}(n)$ and $X_{max}(n)$ are interesting to analyze. For example, making parallels with graph theory, $X_{max}(n)$ is an analogue of the diameter, while $X_{\delta}(n)$ is similar to the average shortest path length~\cite{boccaletti2006complex,newman2003structure} which is often studied as a more stable analogue of the diameter.

For $X_{max}(n)$ the following simple lemma holds.
\begin{lemma}\label{lem:max}
Let $\omega(n)$ be any function tending to infinity as $n$ tends to infinity. Then 
$$
\Prob\left(X_{max}(n) > \frac {n^{\frac 1 \beta}} {\omega(n)} \right) = 1 - o(1)\,.
$$
\end{lemma}

\begin{proof}

Let $r = \frac{n^{\frac 1 \beta}}{\omega(n)}.$ If for at least one step $t$ we have $\xi_t > 2r$, then $X_{max}(n)>r$, therefore
$$
\Prob(X_{max}(n)\le r) \le \Prob(\xi_t \le r\text{ for all $1\le t \le n$}) = \left(1 - \frac{1}{(1+r)^{\beta}} \right)^n \\
\sim e^{\frac{-n}{(1+r)^\beta}} = o(1)\,.
$$

\end{proof}

As a result, we obtain that ``on average'' preferential placement is well concentrated near the origin ($X_{\delta}(n)$ grows at most logarithmically). However, according to Lemma~\ref{lem:max}, the diameter of the obtained configuration grows at least as $n^{1/\beta}$, i.e., it is highly affected by outliers.

\section{Graph models}\label{sec:graph}

\subsection{Possible definitions}

In this section, we discuss how a graph can be constructed based on the vertex embedding produced by the preferential placement procedure.

\paragraph{Spatial distance.} The basic idea behind many known spatial models is that we want to increase the probability of connecting two vertices if they have similar latent features. Various methods can be found in the literature, which are usually combined with some other ideas like introducing weights of vertices or taking into account degrees of vertices (see, e.g., a survey of spatial models in \cite{barthelemy2011spatial}). We now briefly describe some possible approaches:

\begin{itemize}
\item  \emph{threshold model}~\cite{bradonjic2008structure,masuda2005geographical}: 
\begin{equation}\label{eq:threshold}
\Prob\big( (v_i,v_j) \in E \big) = I\big[ \| \x_i - \x_j \| \leq \theta \big]\,,\,\,\,\,\,\,\theta >0\,;
\end{equation}
\item  \emph{$p$-threshold model}:  $$
\Prob\big( (v_i,v_j) \in E \big) = p I\big[ \| \x_i - \x_j \| \leq \theta \big]\,,\,\,\,\,\,0<p<1,\,\theta>0\,;
$$
\item  \emph{$p$-threshold model with random edges} (as in spatial small-world models~\cite{barthelemy2011spatial}): $$
\Prob\big( (v_i,v_j) \in E \big) = p_0 + p_1  I\big[ \| \x_i - \x_j \| \leq \theta \big] \,,\,\,\,\,\,0<p_0,p_1,p_0+p_1<1,\,\theta>0\,;
$$
\item  \emph{inverted distance model}: $$
\Prob\big( (v_i,v_j) \in E \big) \propto \frac{1}{ \| \x_i - \x_j \| }\,;
$$
\item  \emph{Waxman model}~\cite{waxman1988routing}: $$
P\big( (v_i,v_j) \in E \big) \propto e^{-\| \x_i - \x_j \| / d }\,, \,\,\,\,\, d>0\,.
$$
\end{itemize}
Here we denote by $E$ the set of edges and by $I[\cdot]$ the indicator function. 
We assume that all edges are mutually independent, hence to describe a random graph it is enough to define the probability of each edge. 
In Section~\ref{sec:case_study} we present some empirical analysis for the basic threshold model as well as for its modification based also on a genealogical distance which we define in the next paragraph.


\paragraph{Genealogical distance.} It is also reasonable to take into account the genealogical tree used to create the configuration of vertices (see Section~\ref{sec:model_tree} for the description of this tree). For example, it seems very natural to connect each vertex to its parent in the genealogical tree. This serves several purposes: (i) the underlying genealogical process is reflected in the obtained graph, (ii) the graph is guaranteed to be connected, (iii) the graph has small diameter, as we show further in this section. 

More generally, let $d_{tree}(v_i,v_j)$ be the length of the unique path between $v_i$ and $v_j$ in the genealogical tree. Then let 
\begin{equation}\label{eq:general_model}
\Prob\big( (v_i,v_j) \in E \big) = F(\|\x_i-\x_j\|,d_{tree}(v_i,v_j))\,,
\end{equation}
where $F(x,y): \mathbb{R}_{+} \times \mathbb{N} \to [0,1]$ can be any function non-increasing with $x$ for fixed $y$ and with $y$ for fixed $x$.

In Section~\ref{sec:case_study}, in addition to the threshold model~\eqref{eq:threshold}, we also analyze the following \textit{genealogical threshold model}, which is a particular case of~\eqref{eq:general_model}: 
\begin{equation}\label{eq:threshold_distance}
\Prob\big( (v_i,v_j) \in E \big) = I\big[ \| \x_i - \x_j \| \leq \theta \text{ or } d_{tree}(v_i,v_j)) = 1\big]\,.
\end{equation}

\subsection{Case study}\label{sec:case_study}

\begin{figure*}[t]
        \centering
        \begin{subfigure}[b]{0.49\textwidth}
            \centering
          \includegraphics[width=\textwidth]{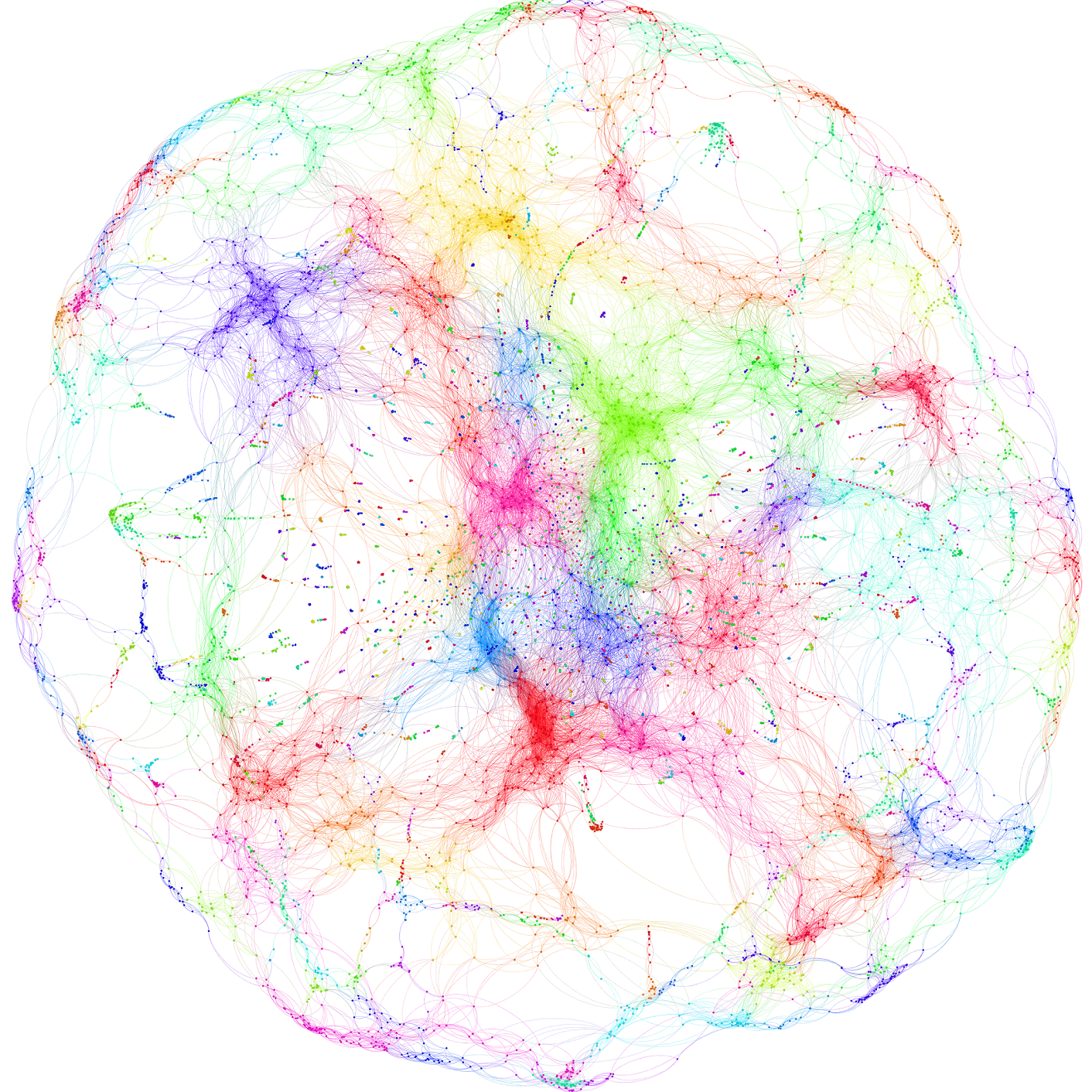}
        \caption{Threshold model}
        \end{subfigure}
        \begin{subfigure}[b]{0.49\textwidth}
            \centering
        \includegraphics[width=\textwidth]{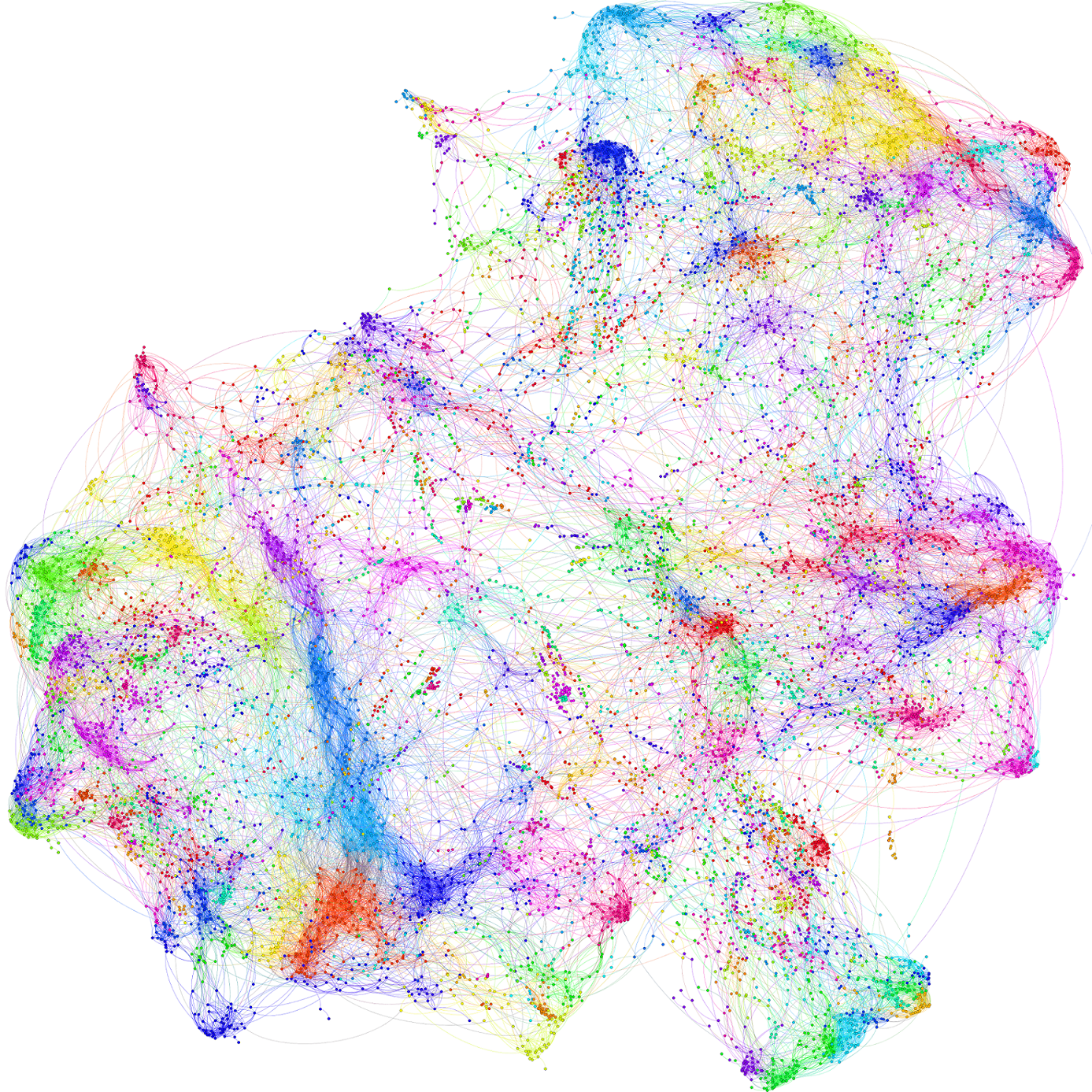}
         \caption{Genealogical threshold model}
        \end{subfigure}
\caption{Visualization of the obtained graphs, laid out with ForceAtlas 2~\cite{jacomy2014forceatlas2}; vertices are colored according to a partition obtained by label propagation community detection algorithm (LPA)~\cite{raghavan2007near}.}\label{fig:visualization}
\end{figure*}

As we promised above, in this section we analyze two models: the threshold model and the genealogical threshold model. In both cases, we choose $\theta$ such that we have $5n$ edges with $\| \x_i - \x_j \| \leq \theta$ in our graph.
As before, we take $\Xi$ to be a distribution with the density function $f_\beta(x) = \frac \beta {(x + 1)^{\beta+1}}, x\geq 0$ for $\beta = 1.5$. 

\paragraph{Visualization.}
Figure~\ref{fig:visualization} visualizes the obtained graphs with $n = 10^4$ (note that other empirical results are obtained for $n=10^5$, as usual). In both cases, one can clearly see communities of various sizes. One difference between two graphs is that in the threshold model the graph is not connected, as expected: 16\% of vertices are isolated and the giant connected component consists of 54\% of vertices.




\paragraph{Degree distribution.}
Let us empirically analyze the degree distribution for both models~\eqref{eq:threshold} and~\eqref{eq:threshold_distance}. The cumulative degree distributions are presented on Figure~\ref{fig:degree_distr}.
Observe that these distributions do not follow a power law. However, they are very similar to degree distributions obtained in many real-world networks (numerous examples can be found in~\cite{real_degree_distribution}).

\begin{figure*}[t]
            \centering
        \begin{subfigure}[b]{0.49\textwidth}
            \centering
          \includegraphics[width=\textwidth]{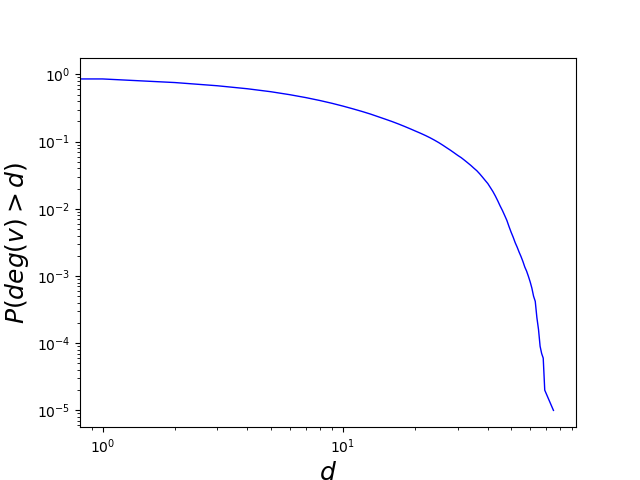}
        \caption{Threshold model}
        \end{subfigure}     
                \begin{subfigure}[b]{0.49\textwidth}
            \centering
          \includegraphics[width=\textwidth]{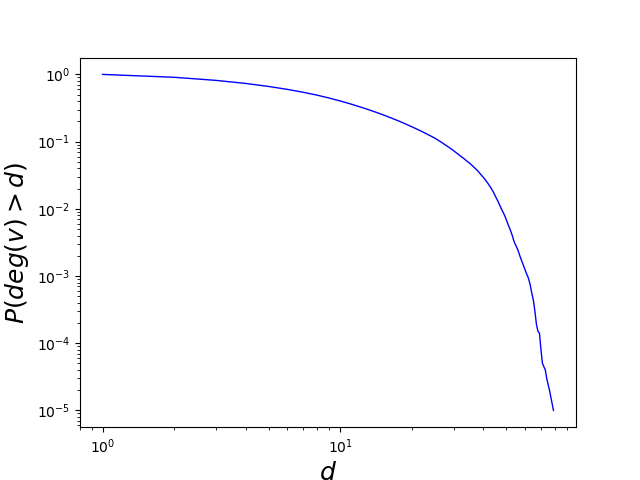}
        \caption{Genealogical threshold model}
        \end{subfigure}
         \caption{Cumulative degree distribution.}\label{fig:degree_distr}
\end{figure*}

\paragraph{Distribution of community sizes.}
Let us also empirically analyze the distribution of community sizes, which is the main focus of the current research. 
To do this, we partition the vertices using a well-known label propagation community detection algorithm (LPA) first proposed in~\cite{raghavan2007near}. The cumulative community size distributions are presented on Figure~\ref{fig:com_distr}. 
Note that both distributions can be approximated by a power law. However, such approximation gives a better fit for the genealogical threshold model.

\begin{figure*}[t]
            \centering
        \begin{subfigure}[b]{0.49\textwidth}
            \centering
          \includegraphics[width=\textwidth]{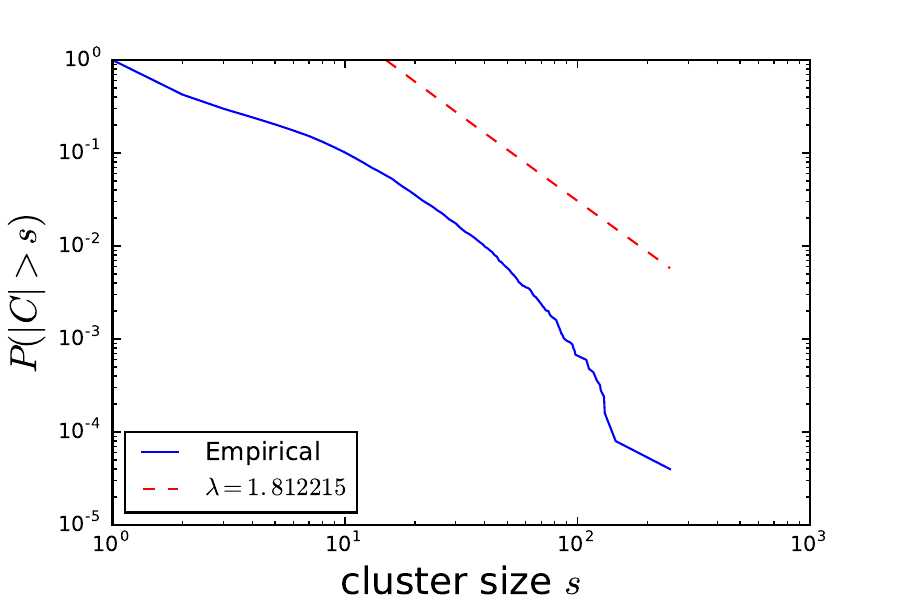}
        \caption{Threshold model}
        \end{subfigure}     
                \begin{subfigure}[b]{0.49\textwidth}
            \centering
          \includegraphics[width=\textwidth]{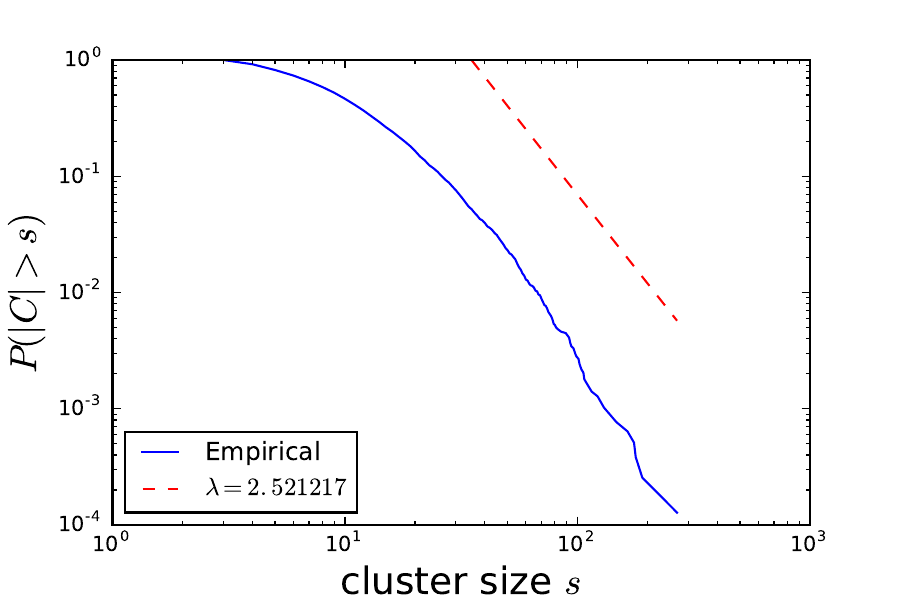}
        \caption{Genealogical threshold model}
        \end{subfigure}
         \caption{Cumulative distribution of community sizes.}\label{fig:com_distr}
\end{figure*}

\paragraph{Diameter.}
Now let us show that graphs obtained according to the genealogical threshold model~\eqref{eq:threshold_distance} have small (at most logarithmic) diameter.
In order to do this, we apply the following theorem~\cite{pittel1994note}.

\begin{theorem}
Let $H_n$ be the height of a uniform random recursive tree $T_n$. Then, with probability one,
$$
\lim_{n \to \infty} \frac{H_n}{\log n} = e\,.
$$
\end{theorem}
It remains to note that our graph contains the underlying genealogical tree by the definition, and this tree is URRT, so we obtain the following result.

\begin{theorem}
For a graph $G_n$ constructed according to~\eqref{eq:threshold_distance} we have, with probability one,
$$
\lim_{n\to \infty}\frac{\mathbf{diam}(G_n)}{\log n} \le 2 e\,.
$$
\end{theorem}

\section{Conclusion and future work}

In this paper, we introduced a principle called \textit{preferential placement}.
Our method is designed to model a realistic clustering structure in $\mathbb{R}^d$. 
The algorithm is parametrized only by a distribution $\Xi$, and if $\Xi$ is a Pareto distribution, which is a natural choice, then we essentially have only one parameter~--- the exponent $\beta$. 
The proposed algorithm naturally models clusters and the distribution of cluster sizes follows a power law, which is a desirable property. Although preferential placement only generates the coordinates of vertices, which is interesting in its own right, one can easily construct a graph based on the obtained structure using one of the methods discussed in this paper. In particular, we empirically analyzed the threshold model and its genealogical-based variant  and obtained graphs with power-law distributions of community sizes and with degree distributions similar to ones observed in many real-worlds networks.

In this paper, we made only a first step to understanding the cluster formation in complex structures and there are many directions for future research. First of all, more formal analysis of the distribution of cluster sizes would be useful. As we discussed, the main problem here is the lack of any suitable formal definition of clusters. However, one can try, e.g., to analyze clusters produced by one of well-known clustering algorithms. Second, the radius of obtained configuration can be analyzed deeper: e.g., we are going to study the lower bound for $X_{\delta}(n)$ and the upper bound for $X_{max}(n)$ (see Section~\ref{sec:radius}).
Finally, there are many open questions in the analysis of the obtained graph structures: comparison of different models, theoretical analysis of the degree distribution, analysis of the local clustering coefficient, and so on.

\section*{Acknowledgements}

This study was funded by RFBR according to the research project 18-31-00207
and by Russian President grant 
MK-527.2017.1. 

\bibliographystyle{plain}
\bibliography{references_clusters}

\section*{Appendix}

\section*{Proof of Theorem~\ref{th:cluster_sizes}}

First, recall the process of cluster formation:
\begin{itemize}
\item At the beginning of the process we have one vertex which forms one cluster;
\item At step $t+1$ with probability $p(t)$ a new cluster consisting of $v_{t+1}$ is created;
\item With probability $1-p(t)$ the new vertex $v_{t+1}$ joins some already existing cluster $C$ with probability proportional to $|C|$.
\end{itemize}
So, we can write the following equations:
\begin{equation}\label{eq:s=1}
\E(F_{t+1}(1)|S_t) = F_{t}(1) \left(1 - \frac{1-p(t)}{t}\right) + p(t)\,,
\end{equation}
\begin{equation}\label{eq:s>1}
\E(F_{t+1}(s)|S_t) = F_{t}(s) \left(1 - \frac{s(1-p(t))}{t}\right) + F_t(s-1)\frac{(s-1)(1-p(t))}{t}\,, \,\,\, \, s>1\,.
\end{equation}
Now we can take expectations of the both sides of the above equations and analyze the behavior of $\E F_{t}(s)$ inductively.
Recall that the following form of the probability $p(t)$ is assumed: 
$p(t) = \frac{c}{t^{\alpha}}$, where $c>0$, $1\le \alpha \le 1$.


\paragraph{Proof for $\alpha=0.$}

Consider the case $\alpha=0$, i.e., $p(t) = c$, $0<c<1$. Let us prove that in this case 
\begin{equation}\label{eq:1}
\E F_n(s) = \frac{c (s-1)!\,\G\left(2+\frac{1}{1-c}\right)}{(2-c)\G\left(s+1+\frac{1}{1-c}\right)}\left(n + \theta_{n,s}\right)\,.
\end{equation}
where $|\theta_{n,s}| \le C \, s^\frac{1}{1-c}$ for some constant $C>0$.

We prove this result by induction on $s$ and for each $s$ the proof is by induction on $n$. Note that for $n=1$ Equation~\eqref{eq:1} holds for all $s$. Indeed, we can take $\theta_{1,s} = -1$ for $s > 1$ and take $\theta_{1,1}$ satisfying $\frac{c}{2-c}\left(1 + \theta_{1,1}\right) = 1$.

Consider now the case $s = 1$. We want to prove that
$$
\E F_n(1) = \frac{c}{2-c}\left(n + \theta_{n,1}\right)\,.
$$
For the inductive step we use Equation~\eqref{eq:s=1} and get
\begin{multline*}
\E(F_{t+1}(1)) = \E F_{t}(1) \left(1 - \frac{1-c}{t}\right) + c 
= \frac{c}{2-c} \left(t + \theta_{t,1}\right) \left(1 - \frac{1-c}{t}\right) + c \\
= \frac{c}{2-c} \left(t+1 + \theta_{t,1} \left(1 - \frac{1-c}{t}\right) \right)\,.
\end{multline*}
Since
$$
C \left(1 - \frac{1-c}{t} \right) \le C,
$$
this finishes the proof for $\alpha=0$ and $s=1$.

For $s>1$ we use Equation~\eqref{eq:s>1} and get
\begin{multline*}
\E(F_{t+1}(s)) = \E F_{t}(s) \left(1 - \frac{s\left(1-c\right)}{t}\right) + \E F_t(s-1)\frac{(s-1)\left(1-c\right)}{t} \\
=  \frac{c (s-1)!\,\G\left(2+\frac{1}{1-c}\right)\left(t + \theta_{t,s}\right)}{(2-c)\,\G\left(s+1+\frac{1}{1-c}\right)}  \left(1 - \frac{s(1-c)}{t}\right) 
 +  \frac{c (s-1)!\,\G\left(2+\frac{1}{1-c}\right)(1-c)(t + \theta_{t,s-1})}{(2-c)\,\G\left(s+\frac{1}{1-c}\right)t} \\
=\frac{c (s-1)!\,\G\left(2+\frac{1}{1-c}\right)}{(2-c)\,\G\left(s+1+\frac{1}{1-c}\right)} 
\left(t + 1 + \theta_{t,s} \left(1 - \frac{s(1-c)}{t}\right) + \theta_{t,s-1} \frac{s(1-c) + 1}{t} \right)\,.
\end{multline*}
To finish the proof we need to show that
\begin{equation}\label{eq:c}
(s-1)^{\frac{1}{1-c}} \frac{s(1-c) + 1}{t}  \le s^{\frac{1}{1-c}}\frac{s(1-c)}{t}\,.
\end{equation}
It is easy to show that the above inequality holds. Indeed,~\eqref{eq:c} can be rewritten as
$$
\left(1 - \frac{1}{s-1}\right)^{\frac{1}{1-c}} \ge
  1 + \frac{1}{s(1-c)} \,.
$$
which holds as $\left(1 - \frac{1}{s-1}\right)^{\frac{1}{1-c}} \ge 1 - \frac{1}{(s-1)(1-c)} > 1 - \frac{1}{s(1-c)}$.


\paragraph{Proof for $0<\alpha \le 1$.}

Now we consider the case $p(t) = ct^{-\alpha}$ for $0< \alpha \le 1$. Let us prove that in this case 
$$
\E F_n(s) = \frac{c (s-1)! \, \G(3-\alpha)}{(2-\alpha)\G(s+2-\alpha)} \left(n^{1-\alpha} + \theta_{n,s}\right)\,,
$$
where $|\theta_{n,s}| \le C n^{\max\{0,1-2\alpha\}}s^{1-\alpha+\epsilon}$
for some constant $C>0$ and for any $\epsilon>0$.

The proof is similar to the case $\alpha = 0$.
Again, for $n=1$ the theorem holds, as we can take $\theta_{1,s} = -1$ for $s > 1$ and take $\theta_{1,1}$ satisfying $\frac{c}{2-\alpha}\left(1 + \theta_{1,1}\right) = 1$.

Consider the case $s = 1$. We want to prove that
$$
\E F_n(1) = \frac{c}{2-\alpha}  \left(n^{1-\alpha} + \theta_{n,1}\right).
$$
Inductive step in this case becomes
\begin{multline*}
\E(F_{t+1}(1)) = \E F_{t}(1) \left(1 - \frac{1-ct^{-\alpha}}{t}\right) + ct^{-\alpha}  
= \frac{c}{2-\alpha} \left(t^{1-\alpha} + \theta_{t,1}\right) \left(1 - \frac{1-ct^{-\alpha}}{t}\right) + ct^{-\alpha} \\
= \frac{c}{2-\alpha} \left( t^{1-\alpha} - t^{-\alpha} + c\,t^{-2\alpha} + (2-\alpha)t^{-\alpha} +\theta_{t,1} \left(1 - \frac{1-ct^{-\alpha}}{t}\right) \right) \\
= \frac{c}{2-\alpha} \left( (t+1)^{1-\alpha} + O\left(t^{-\alpha-1}\right) + c\,t^{-2\alpha} +\theta_{t,1} \left(1 - \frac{1-ct^{-\alpha}}{t}\right) \right)\,.
\end{multline*}
The last equation holds since $(t+1)^{1-\alpha} = t^{1-\alpha} + (1-\alpha)t^{-\alpha} + O\left(t^{-\alpha-1}\right)$.
In order to finish the proof for the case $s=1$ it is sufficient to show that
$$
O\left(t^{-\alpha-1}\right) + c\, t^{-2\alpha}  \le C t^{\max\{0,1-2\alpha\}} \frac{1-ct^{-\alpha}}{t}\,,
$$
which holds for sufficiently large $C$, since
$$
C t^{\max\{0,1-2\alpha\}} \frac{1-ct^{-\alpha}}{t} \ge C t^{-2\alpha} + O\left(C t^{-3\alpha} \right) \ge   c\, t^{-2\alpha} + O\left(t^{-\alpha-1}\right) 
$$
for all $\alpha$, $0 < \alpha \le 1$, and for sufficiently large $C$.

For $s>1$ we have:
\begin{multline*}
\E(F_{t+1}(s)) = \E F_{t}(s) \left(1 - \frac{s\left(1-ct^{-\alpha}\right)}{t}\right) + \E F_t(s-1)\frac{(s-1)\left(1-ct^{-\alpha}\right)}{t} \\
= \frac{c (s-1)!\, \G(3-\alpha)}{(2-\alpha)\G(s+2-\alpha)} \left(t^{1-\alpha} + \theta_{t,s}\right)\left(1 - \frac{s\left(1-ct^{-\alpha}\right)}{t}\right) \\ + \frac{c (s-2)!\, \G(3-\alpha)}{(2-\alpha)\G(s+1-\alpha)} \left(t^{1-\alpha} + \theta_{t,s-1}\right) \frac{(s-1)\left(1-ct^{-\alpha}\right)}{t} \\
= \frac{c (s-1)!\, \G(3-\alpha)}{(2-\alpha)\G(s+2-\alpha)} \left( t^{1-\alpha}- s\left(1-ct^{-\alpha}\right)t^{-\alpha} + \theta_{t,s}\left(1 - \frac{s\left(1-ct^{-\alpha}\right)}{t}\right) \right. \\ \left. +(s+1-\alpha)\left(1-ct^{-\alpha}\right)t^{-\alpha} + \theta_{t,s-1} \frac{(s+1-\alpha)\left(1-ct^{-\alpha}\right)}{t} \right)\\
= \frac{c (s-1)!\, \G(3-\alpha)}{(2-\alpha)\G(s+2-\alpha)} \bigg( (t+1)^{1-\alpha} + O\left(t^{-\alpha-1} \right) - c (1 - \alpha) t^{-2\alpha} \\ \left.  + \theta_{t,s}\left(1 - \frac{s\left(1-ct^{-\alpha}\right)}{t}\right) +   \theta_{t,s-1} \frac{(s+1-\alpha)\left(1-ct^{-\alpha}\right)}{t}  \right)\,.
\end{multline*}

Recall the condition for $\theta_{n,s}$: $|\theta_{n,s}| \le C n^{\max\{0,1-2\alpha\}}s^{1-\alpha+\epsilon}$. Therefore, in order to finish the proof, 
we need to show that for some $C$ the sum of the error terms above can be bounded by $C (t+1)^{\max\{0,1-2\alpha\}}s^{1-\alpha+\epsilon}$, i.e.,
\begin{multline*}
O\left(t^{-\alpha-1} \right) - c (1 - \alpha) t^{-2\alpha}  + C t^{\max\{0,1-2\alpha\}}s^{1-\alpha+\epsilon}\left(1 - \frac{s\left(1-ct^{-\alpha}\right)}{t}\right) \\ +   C t^{\max\{0,1-2\alpha\}}(s-1)^{1-\alpha+\epsilon} \frac{(s+1-\alpha)\left(1-ct^{-\alpha}\right)}{t} \le
C (t+1)^{\max\{0,1-2\alpha\}}s^{1-\alpha+\epsilon}\,.
\end{multline*}
Since $(t+1)^{\max\{0,1-2\alpha\}} \ge t^{\max\{0,1-2\alpha\}}$, it is sufficient to prove that 
\begin{multline*}
O\left(t^{-\alpha-1} \right) + c (1 - \alpha) t^{-2\alpha} +   C\, t^{\max\{0,1-2\alpha\}}(s-1)^{1-\alpha+\epsilon} \frac{(s+1-\alpha)\left(1-ct^{-\alpha}\right)}{t} \\ \le C t^{\max\{0,1-2\alpha\}}s^{1-\alpha+\epsilon} \frac{s\left(1-ct^{-\alpha}\right)}{t} \,,
\end{multline*}
which is equivalent to
$$
O\left(t^{-\alpha} \right) + \frac{t^{1-2\alpha} c (1 - \alpha)}{1-ct^{-\alpha}} 
\le C t^{\max\{0,1-2\alpha\}}\left(s^{2-\alpha+\epsilon} - (s+1-\alpha)(s-1)^{1-\alpha+\epsilon} \right) \,.
$$
By using the inequality $(s-1)^{1-\alpha+\epsilon} \le s^{1-\alpha+\epsilon} - (1-\alpha+\epsilon)s^{-\alpha+\epsilon}$, we get
$$
O\left(t^{-\alpha} \right) + \frac{t^{1-2\alpha} c (1 - \alpha)}{1-ct^{-\alpha}} 
\le C t^{\max\{0,1-2\alpha\}} s^{1-\alpha+\epsilon}\left(\epsilon + \frac{(1-\alpha)(1-\alpha+\epsilon)}{s} \right) \,.
$$
Note that the inequality 
$$
O\left(t^{-\alpha} \right) + O\left(t^{1-2\alpha} \right)
\le C t^{\max\{0,1-2\alpha\}} s^{1-\alpha+\epsilon}\epsilon
$$
holds for sufficiently large $C$, which completes the proof of the theorem. 

\section*{Proof of Lemma~\ref{lem:moments}}

Let us first prove that for any vectors $\x,\y\in \mathbb{R}^d$ and for any $\delta$, $1 \le \delta \le 2$, we have:
\begin{equation}\label{eq:vectors}
\frac 1 2 \left(\|\x+\y\|^\delta + \|\x-\y\|^\delta\right) \le \|\x\|^\delta + \|\y\|^\delta\,. 
\end{equation}
If $\|\x\|=0$ and $\|\y\|=0$, then this inequality holds. Otherwise, without loss of generality, we can assume that $\|\x\|^2 + \|\y\|^2 = 1$ (if $\|\x\|^2 + \|\y\|^2 = a > 0$, we divide $\x$ and $\y$ by $\sqrt{a}$).

In order to prove~\eqref{eq:vectors}, it is sufficient to show that  $\frac 1 2 \left(\|\x+\y\|^\delta + \|\x-\y\|^\delta\right) \le 1$ and $1 \le \|\x\|^\delta + \|\y\|^\delta$. The second inequality is trivial:
$$
\|\x\|^\delta + \|\y\|^\delta \ge \|\x\|^2 + \|\y\|^2 = 1
$$
since $\|\x\| \le 1, \|\y\| \le 1$, and $\delta \le 2$.

It remains to show that $\frac 1 2 \left(\|\x+\y\|^\delta + \|\x-\y\|^\delta\right) \le 1$. First let us note that 
$$
\|\x+\y\|^2 + \|\x-\y\|^2 = 2(\|\x\|^2 + \|\y\|^2) = 2.
$$ Now we introduce the following variables 
$$
u = \frac{\|\x+\y\|^2}{2},\,\, v = \frac{\|\x-\y\|^2}{2}, \,\,\,
 u+v = 1
$$ 
and it remains to show that $2^{\delta/2-1}\left(u^{\delta/2}+v^{\delta/2}\right) \le 1$, which is true since the maximum of $u^{\delta/2}+v^{\delta/2} = u^{\delta/2}+(1-u)^{\delta/2}$ for $0 \le u \le 1$ is equal to $2^{-\delta/2+1}$ at $u = \frac 1 2$.

Now we note that the distributions of $\bxi + \boldeta$ and $\bxi - \boldeta$ are identical, so 
$$
\E\|\bxi + \boldeta\|^\delta = \frac{1}{2} \left(\E\|\bxi + \boldeta\|^\delta + \E\|\bxi - \boldeta\|^\delta\right)
$$
and now the lemma follows from Equation~\eqref{eq:vectors}.

\end{document}